\documentclass[preprint,authoryear,review,12pt]{elsarticle}
\usepackage{lineno,hyperref} 
\modulolinenumbers[5]

\journal{Biosystems Engineering}

\usepackage{amsmath}
\usepackage{amssymb} 

\usepackage{tikz}
\usetikzlibrary{arrows,matrix,positioning,patterns}
\usetikzlibrary{calc}
\usepackage{pgfplots} 
\usetikzlibrary{external}
\usetikzlibrary{plotmarks}

\newdefinition{rmk}{Remark}
\newproof{pf}{Proof}
\newproof{pot}{Proof of Theorem \ref{thm2}}

\usepackage{algorithm} 
\usepackage{algpseudocode}

\usepackage{amsthm}

\newtheorem{proposition}{Proposition}

\newtheorem{remark}{Remark}

\usepackage{amsmath}
\usepackage{amssymb} 

\usepackage{eurosym}
\usepackage{colortbl}

\usetikzlibrary{arrows.meta,calc,decorations.markings,math,arrows.meta}

\usepackage[labelformat=simple]{subcaption}
\DeclareCaptionLabelSeparator{periodspace}{.\quad}
\usepackage{color,framed}
\usepackage[utf8]{inputenc} 
\usepackage{tabularx}

\usepackage{rotating}
\usepackage{graphicx}
\usepackage{lscape}
\allowdisplaybreaks 
\usepackage{longtable}

\usepackage[utf8]{inputenc}
\usepackage{pgfplots}
\usepgfplotslibrary{groupplots} 
\pgfplotsset{compat=1.3}

\begin{document}

\begin{frontmatter}
\title{Partial Field Coverage Based on Two Path Planning Patterns}

\author[label1]{Mogens Graf Plessen\corref{cor1}}
\ead{mogens.plessen@imtlucca.it}
\address[label1]{IMT School for Advanced Studies Lucca, Piazza S. Francesco, 19-55100 Lucca, Italy}
\cortext[cor1]{Corresponding author.}

\begin{abstract}
This paper presents a path planning method for partial field coverage. Therefore, a specific path planning pattern is proposed. The notion is that lighter machinery with smaller storage tanks can alleviate soil compaction because of a reduced weight, but does not enable full field coverage in a single run because of the smaller storage capacity. This is relevant for spraying applications and related in-field work. Consequently, multiple returns to a mobile or stationary depot located outside of the field are required for storage tank refilling. Therefore, a suitable path planning method is suggested that accounts for the limited turning radii of agricultural vehicles, satisfies compacted area minimisation constraints, and aims at overall path length minimisation. The benefits of the proposed method are illustrated by means of a comparison to a planning method based on the more common AB pattern. It is illustrated how the proposed path planning pattern can also be employed efficiently for single-run field coverage. 
\end{abstract} 
\end{frontmatter}


 
\begin{small}
\fcolorbox{black}{white}{
\parbox{0.8\textwidth}{ 
\begin{tabular}{ll}
\multicolumn{2}{c}{NOMENCLATURE}\\
\multicolumn{2}{l}{Symbols}\\
$D^{(\rho)}$ & Total in-field path length for $\rho\geq 1$ field runs (m)\\
$\Delta D$ & Path length difference (m)\\
$e_{i,j}$ & Edge connecting nodes $i$ and $j$ (m)\\
$f(t)$ & Storage tank fill-level ($\%$)\\
$\gamma(t)$ & Vehicle state (resume, coverage, return)\\
$H_0$ & Nominal lane path length in a rectangular field (m)\\
$N$ & Number of interior lanes (-)\\
$q_l$, $p$ & Auxiliary variables in the example of Section \ref{subsec_parametricEx} (-)\\
$\rho$ & Number of field runs required for field coverage (-)\\
$R$ & Vehicle turning radius (m)\\
$W_0$ & Nominal machine operating width (m)\\
$Z$ & Position $Z=(\xi,\eta)$ (m)\\
$(x,y)$ & Position in the global coordinate system (m)\\
$(\xi,\eta)$ & Position in the normalised coordinate system (m)\\
$Z_0$ & Start position $Z_{0}=(\xi_0,\eta_0)$ (m)\\
$\mathcal{Z}_0^{(l)}$ & Two sets of start positions with $l=1,2$ (m)\\
$Z_i$ & Position of node $i$ (m)\\
$Z(t)$ & Position of agricultural vehicle at time $t$ (m)\\ 
$Z(\tau^\text{last})$ & Position for resuming field coverage (m)\\ 
[3pt] 
\multicolumn{2}{l}{Abbreviations}\\
ABp & Path Planning Method 1 (AB pattern)\\
CIRC & Path Planning Method 2 (circular pattern)\\
CIRC$^\star$ & Path Planning Method 3 (circular pattern)\\
\end{tabular}
}}
\end{small}

\section{Introduction\label{sec_intro}}

According to \cite{ahumada2009application} and \cite{bochtis2010machinery}, the agri-food supply chain can be decomposed into four main functional areas: production, harvesting, storage and distribution. For improved supply chain efficiency, logistical optimisation and route planning play an important role  in all of the four functional areas. Regarding production, for example, by means of minimisation of the non-working distance travelled by machines operating in the headland field according to \cite{bochtis2008minimising}, optimal route planning based on B-patterns according to \cite{bochtis2013benefits}, or route planning for the coordination of fleets of autonomous vehicles as discussed in \cite{conesa2016route} and \cite{seyyedhasani2017using}. See also~\cite{day2011engineering} for an overview of means for efficiency improvements,~\cite{bochtis2013satellite} for the importance of satellite-based navigation systems in modern agriculture, and  \cite{sorensen2010conceptual} for a distinction between in-field, inter-field, inter-sector and inter-regional logistics. The path planning method for partial field coverage presented in this paper relates to the first functional area of the agri-food supply chain. 

The last decades have witnessed a trend towards the employment of larger and more powerful machines in agriculture. This trend is expected to further continue in the near future, see \cite{kutzbach2000trends} and \cite{dain2013risk}. Among the main benefits are higher work rates. The drawbacks include increased soil compaction due to machinery weights, see~\cite{raper2005agricultural} and \cite{hamza2005soil}. See also \cite{antille2013soil} for the influence of tyre sizes on soil compaction. Concurrently to this ongoing trend, there are alternative considerations about the replacement of heavy machinery by teams of smaller and lighter autonomous robots to mitigate soil compaction, see \cite{blackmore2008specification}, \cite{bochtis2010vehicle}, \cite{bochtis2012dss}, \cite{bochtis2013satellite}, \cite{gonzalez2016fleets} and \cite{seyyedhasani2017using}. See also \cite{vougioukas2012distributed} for a method for motion coordination of teams of autonomous agricultural vehicles.

This paper is motivated by the concept of smaller in-field operating machines collaborating with out-field support units (mobile depots). Therefore, a pattern-based path planning method for partial field coverage is presented, which is characterised by i) minimisation of traveled non-working path length, and ii) compliance with compacted area minimisation constraints. The latter implies driving along unique and established transitions between headland path and interior lanes, thereby avoiding the creation of any additional tyre traces that result from vehicle traffic passing over crops and compacting soil. Under the assumption of specific field shapes two different path planning patterns are compared.

In contrast to route planning methods such as in \cite{conesa2016mix} for the in-field operation of a fleet of vehicles, the presented method focuses on the in-field operation of a single vehicle that is repeatedly returning to the field entrance for refilling. This is primarily motivated by the targeted crops (wheat, rapeseed and barley) and the costs of corresponding agricultural vehicles. A support unit, acting as a mobile depot, is assumed to be waiting at the field entrance for refilling. Two comments are made. First, unlike during harvest, mobile units for refilling of spraying tanks cannot come to any arbitrary position along the headland. Second, a single field entrance is in line with the objective of compacted area minimisation. For aforementioned targeted crops, any new field entrance would result in a new compacted area for the connection of in-field headland path and out-field road network, see Fig. \ref{fig_problFormulation_cropped}.

This paper is organised as follows. The problem is formulated in Section \ref{sec_problFormulation}. The main contribution is given in Section \ref{sec_PartialFieldCoverage}. Examples and a discussion are presented in Sections \ref{sec_IllustrativeEx} and \ref{sec_discussion}, before concluding with Section \ref{sec_conclusion}.

\section{Problem Formulation and Notation\label{sec_problFormulation}}

\subsection{Problem Formulation}

\begin{figure}
\vspace{0.1cm}
\centering
\includegraphics[width=9cm]{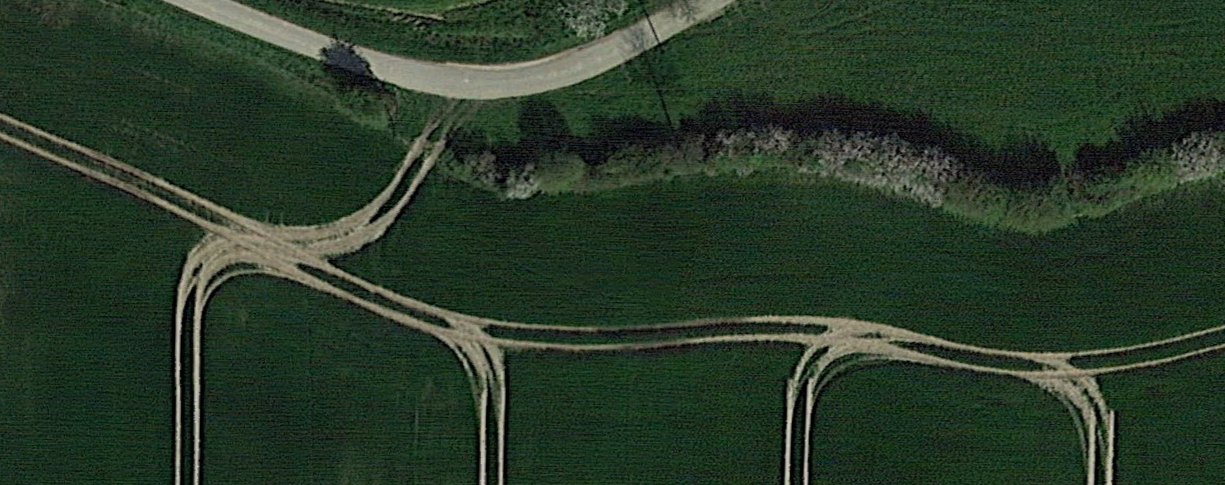}
\caption{Visualisation of real-world transitions between headland path and interior lanes. In the satellite image, the effects of a limited turning radius of the employed agricultural vehicle is visible. The compacted areas are indicated by the bright tyre traces. Note also the compacted area due to the connection between field entrance and headland path.}
\label{fig_problFormulation_cropped}
\end{figure}

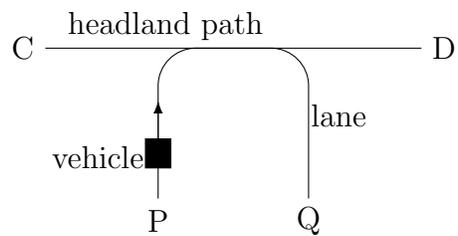
\begin{figure}[t]
\centering
\begin{tikzpicture}
\draw [black] plot [rounded corners=0.5cm] coordinates { (1,0)(1,2)(3,2)(3,0)};
\draw [black] plot [rounded corners=0.5cm] coordinates { (-0.5,2)(4.5,2)};
%
%
\draw [draw=black,draw opacity=1, line width=10pt] plot [rounded corners=0.25cm] coordinates { (1,0.4)(1,0.8)};
\draw [black,-{Latex[scale=1.0]}] plot [rounded corners=0.25cm] coordinates { (1,0.8)(1,1.3)};
\node[color=black] (a) at (3.4, 1.1) {lane};
\node[color=black] (a) at (0.2, 0.55) {vehicle};
\node[color=black] (a) at (1.1, 2.3) {headland path};
\node[color=black] (a) at (1, -0.3) {P};
\node[color=black] (a) at (3, -0.3) {Q};
\node[color=black] (a) at (-0.8, 2) {C};
\node[color=black] (a) at (4.8, 2) {D};
\end{tikzpicture}
\caption{Visualisation of the \emph{compacted area minimisation constraint}.}
\label{fig_Def_LaneHeadlLane}
\end{figure}

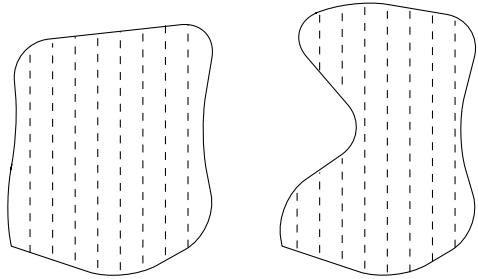
\begin{figure}
\centering
\begin{tikzpicture}
\draw[dashed] (0.25, -0.09) -- (0.25, 2.66);
\draw[dashed] (0.55, -0.1) -- (0.55, 2.7);
\draw[dashed] (0.85, -0.19) -- (0.85, 2.79);
\draw[dashed] (1.15, -0.29) -- (1.15, 2.82);
\draw[dashed] (1.45, -0.32) -- (1.45, 2.86);
\draw[dashed] (1.75, -0.23) -- (1.75, 2.9);
\draw[dashed] (2.05, -0.12) -- (2.05, 2.96);
\draw[dashed] (2.35, -0.02) -- (2.35, 2.95);
\draw[solid,rounded corners=0.5cm] (0, 0) -- (1.5, -0.5) -- (2.75,0.2)--(2.5,1.5)--(2.75,3)--(0,2.7)--(0.1,1.5)--(-0.1,0.6)--(0,0);
\draw[dashed] (3.8, -0.03) -- (3.8, 0.81);
\draw[dashed] (4.1, -0.09) -- (4.1, 0.98);
\draw[dashed] (4.4, -0.17) -- (4.4, 1.17);
\draw[dashed] (4.1, 2.34) -- (4.1, 3.1);
\draw[dashed] (4.4, 2.15) -- (4.4, 3.2);
\draw[dashed] (4.7, -0.3) -- (4.7, 3.2);
\draw[dashed] (5, -0.35) -- (5, 3.2);
\draw[dashed] (5.3, -0.35) -- (5.3, 3.2);
\draw[dashed] (5.6, -0.2) -- (5.6, 3.1);
\draw[dashed] (5.9, -0.02) -- (5.9, 3);
\draw[solid,rounded corners=0.5cm] (3.6, 0) -- (5.1, -0.5) -- (6.28,0.2)--(5.9,1.5)--(6.28,3)--(4.5,3.3)--(3.6,2.9)--(4.8,1.5)--(3.5,0.6)--(3.6,0);
%
\end{tikzpicture}
\caption{Illustration of two types of field shapes with uninterrupted (left) and interruped (right) lanes when aligned in a rotated coordinate frame. Note that field areas do not necessarily have to be convex. For the present paper, the focus is on the left field type.}
\label{fig_2FieldShapes}
\end{figure}

\begin{figure}[t]
\centering
\begin{tikzpicture}
\draw[dashed] (0, 0) -- (0, 4);
\draw[draw=black,fill=black] (0,0) circle (1pt); 
\draw[draw=black,fill=black] (0,4) circle (1pt); 
\draw[dashed] (0.8, 0) -- (0.8, 4);
\draw[draw=black,fill=black] (0.8,0) circle (1pt); 
\draw[draw=black,fill=black] (0.8,4) circle (1pt); 
\draw[dashed] (1.6, 0) -- (1.6, 4);
\draw[draw=black,fill=black] (1.6,0) circle (1pt); 
\draw[draw=black,fill=black] (1.6,4) circle (1pt); 
\draw[dashed] (2.4, 0) -- (2.4, 4);
\draw[draw=black,fill=black] (2.4,0) circle (1pt); 
\draw[draw=black,fill=black] (2.4,4) circle (1pt); 
\draw[dashed] (3.2, 0) -- (3.2, 4);
\draw[draw=black,fill=black] (3.2,0) circle (1pt); 
\draw[draw=black,fill=black] (3.2,4) circle (1pt); 
\draw[dashed] (4.0, 0) -- (4.0, 4);
\draw[draw=black,fill=black] (4.0,0) circle (1pt); 
\draw[draw=black,fill=black] (4.0,4) circle (1pt); 
\draw[dashed] (4.8, 0) -- (4.8, 4);
\draw[draw=black,fill=black] (4.8,0) circle (1pt); 
\draw[draw=black,fill=black] (4.8,4) circle (1pt); 
\draw[draw=black,fill=black] (-0.8,2) circle (1pt); 
\draw[draw=black,fill=black] (5.6,2) circle (1pt); 
\node[color=black] (a) at (-1.3, 2) {$Z_{2N+1}$};
\node[color=black] (a) at (6.15, 2) {$Z_{2N+2}$};
\draw [black,-{Latex[scale=1.0]}] plot [rounded corners=0.25cm] coordinates { (0.4,4)(-0.8,4)(-0.8,0)(5.6,0)(5.6,4)(-0.8,4)(-0.8,0.75)};
\node[color=black] (a) at (0, -0.35) {$Z_1$};
\node[color=black] (a) at (0.8, -0.35) {$Z_2$};
\node[color=black] (a) at (4.8, -0.35) {$Z_N$};
\node[color=black] (a) at (0, 4.3) {$Z_{N+1}$};
\node[color=black] (a) at (4.8, 4.3) {$Z_{2N}$};
\draw[draw=black,fill=black] (1.2,4) circle (1pt); 
\node[color=black] (a) at (1.2, 4.3) {$Z_{0}$};
\draw [black,-{Latex[scale=1.0]}] plot [rounded corners=0.25cm] coordinates { (-3,0)(-2,0)};
\draw [black,-{Latex[scale=1.0]}] plot [rounded corners=0.25cm] coordinates { (-2.7,-0.3)(-2.7,0.7)};
\node[color=black] (a) at (-1.85, -0.3) {$\xi$};
\node[color=black] (a) at (-3.1, 0.7) {$\eta$};
%
\end{tikzpicture}
\caption{Illustration of notation. The headland path is denoted by the solid line. Interior lanes are here indicated by the dashed lines. The arrow indicates the headland traversal direction. Headland and lanes are expressed in the $(\xi,\eta)$-coordinate system such that lanes are aligned to the $\eta$-axis. The field entrance (start position) is denoted by $Z_0$.}
\label{fig_IllustrNotation}
\end{figure}
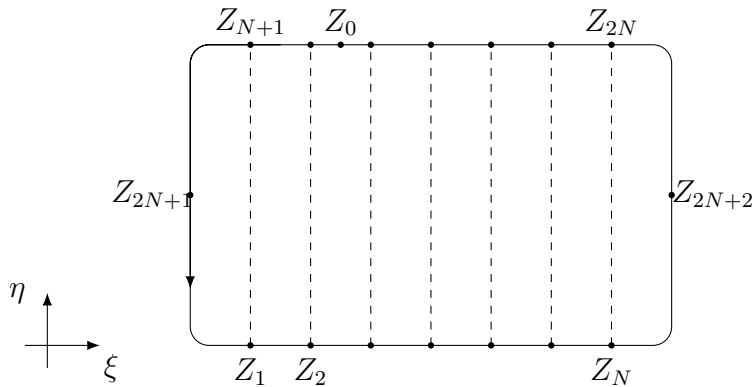

This paper addresses pattern-based path planning for partial field coverage. The fundamental objective is non-working path length minimisation. Therefore, the following is additionally addressed.

First, path planning must account for \emph{compacted area minimisation constraints}. These constraints impose unique transitions between headland and interior lanes and account for limited turning radii. For illustration see Fig. \ref{fig_Def_LaneHeadlLane}. Any agricultural vehicle that is traveling along lanes and the headland path must respect tractor traces established upon first field coverage. Thus, transitions P-Q and P-D are admissible. In contrast, transition P-C is not admissible.  Such a transition would deviate from established tyre traces when accounting for the limited turning radius of the vehicle, and would therefore repress and destroy precious crop. See also \cite{plessen2016shortest}~for general shortest path in-field navigation accounting for these constraints. 

Second, path planning must ideally minimise non-working path length for \emph{both} single-run and partial field coverage. 

Third, path planning must optimally account for the following tasks during online operation: i) path following according to a field coverage plan, ii) navigation from a position along the path network to the field entrance for refilling of storage tanks, and iii) navigation from the field entrance after refilling back to the position along the field coverage path for the resumption of work.

\subsection{Assumptions on Field Shapes and Notation}

The focus of this paper is on field shapes that permit optimal path planning based on patterns, see Fig. \ref{fig_2FieldShapes}. As will be shown, for these field shapes the preferred pattern-based path planning method can yield minimal path length solutions for both full and partial field coverage. Relevant components for planning include a headland path and multiple interior lanes, see Fig. \ref{fig_IllustrNotation}. In combination, they enable field coverage. The headland path is generated from an erosion (mathematical operation) of the field contour to the field interior. All position coordinates are initially expressed in the global $(x,y)$-coordinate system. Then, all coordinates can be transformed by a rotation of angle $\theta$ to a rotated $(x_\theta,y_\theta)$-coordinate system such that interior lanes are aligned with the vertical $y_\theta$-axis. Thus, $(x_\theta,y_\theta) = R(\theta)(x,y)$, where $R(\theta)$ denotes a standard rotation matrix with rotation angle $\theta$. In addition, at most two coordinate reflections (Householder transformation in two dimensions) are employed in order to normalise the path planning problem with respect to the field entrance position. Thus, the normalised coordinate system in which paths are ultimately planned is described by coordinates $(\xi,\eta)$. Mathematical details about the coordinate reflections follow in Section \ref{subsec_MEAND}. All three transformation steps (rotation and at most two reflections) are linear. Ultimately, paths planned in the $(\xi,\eta)$-system are recovered in the $(x,y)$-plane by inversion of the linear transformations. In the following, a position within the normalised coordinate system is abbreviated by
\begin{equation}
Z = (\xi,\eta).
\end{equation}
Field contours are assumed such that any rotated interior lane is uninterrupted. Thus, it can be represented by one continuous lane segment for a given $\xi$-coordinate. This assumption is made to enable optimal path planning based on patterns. In general, aforementioned interruptions may arise from deep field indents, bays or tree islands. Note that the previous discussion can be generalised to interior lanes that are curvilinear to a particular part of the field contour. 

The following additional assumptions are made. First, the orientation $\theta$ of the interior lanes within the global coordinate system are assumed to be given. However, the transitions between headland path and interior lanes (left or right turns) are initially not specified. They result from the presented methods. Second, besides the headland path and interior lanes, a field entry located along the headland path is assumed. It is additionally assumed that a designated field exit exists, which may be identical to the field entrance.

See Fig. \ref{fig_IllustrNotation} for illustration of notation and the labeling of transition points between headland path and interior lanes. Not displayed are i) the field exit, which is labeled by $Z_{2N+3}$ and interchangedly abbreviated by $Z_e$ for brevity, and which may be identical to $Z_0$, and ii) the agricultural vehicle position at time $t$ which is labeled by $Z_{2N+4}$ or interchangedly by $Z(t)$, and which may be located anywhere along the headland path or along any interior lane. Note that the precise transitions between headland and interior lanes that account for limited turning radii are initially not specified. Locations $Z_i,~\forall i=0,\dots,2N+4$, denote nodes. Their path connections define edges $e_{i,j},\forall i,j=0,\dots,2N+4$, whereby any edge weight is specified by its path length. Combining edges and nodes, a transition graph can be generated. Nodes indexed by $i=2N+1$ and $i=2N+2$ are introduced to ensure unique edge connections between any two nodes. Based on the transition graph, shortest paths can be determined (\cite{bertsekas1995dynamic}).

\section{Path Planning for Partial Field Coverage\label{sec_PartialFieldCoverage}}

This section discusses three pattern-based path planning methods for partial field coverage: ABp, CIRC and CIRC$^\star$.

\subsection{Path Planning based on the AB pattern -- ABp\label{subsec_MEAND}}

\begin{figure}
\vspace{0.3cm}
\begin{subfigure}[b]{\linewidth}
\centering%
\begin{tikzpicture}
\draw [black,-{Latex[scale=1.0]}] plot [rounded corners=0.25cm] coordinates { (0.8,2)(-0.8,2)(-0.8,0)(5.6,0)(5.6,2)(0.45,2)
(0,1.9)(0,0.1)(0.8,0.1)(0.8,1.9)(1.6,1.9)(1.6,0.1)(2.4,0.1)(2.4,1.9)(3.2,1.9)(3.2,0.1)(4.0,0.1)(4.0,1.9)(4.8,1.9)(4.8,0.1)(4.4,0.1)
};
\draw[dotted,blue,line width=1pt] (2.4, -0.2) -- (2.4, 2.2);
\draw [black,-{Latex[scale=1.0]}] plot [rounded corners=0.25cm] coordinates { (-0.8,1.5)(-0.8,0.5)
};
\draw [black,-{Latex[scale=1.0]}] plot [rounded corners=0.25cm] coordinates { (4.8,1)(4.8,0.1)(5.2,0.1)};
\draw [draw=red,draw opacity=0.5, line width=4pt] plot [rounded corners=0.25cm] coordinates { (2.4,2)(0.45,2)
};
\node[color=red] (a) at (1.25, 2.4) {$\mathcal{Z}_0^{(1)}$};
\node[color=black] (a) at (4.8, -0.35) {towards $Z_e$};
\end{tikzpicture}
\caption{Odd $N$ and $\mathcal{Z}_0^{(1)}$.}
\end{subfigure}\\[10pt]
\begin{subfigure}[b]{\linewidth}
\centering%
\begin{tikzpicture}
\draw [black,-{Latex[scale=1.0]}] plot [rounded corners=0.25cm] coordinates { (-0.8,-5.75)(-0.8,-6)(5.6,-6)(5.6,-4)(-0.8,-4)(-0.8,-6)(0,-5.95)(0,-4.1)(0.8,-4.1)(0.8,-5.9)(1.6,-5.9)(1.6,-4.1)(2.4,-4.1)(2.4,-5.9)(3.2,-5.9)(3.2,-4.1)(4.0,-4.1)(4.0,-5.9)(4.8,-5.9)(4.8,-4.1)(4.4,-4.1)
};
\draw [black,-{Latex[scale=1.0]}] plot [rounded corners=0.25cm] coordinates { (4.8,-5)(4.8,-4.1)(5.2,-4.1)
}; 
\draw[dotted,blue,line width=1pt] (2.4, -6.2) -- (2.4, -3.8);
\draw [black,-{Latex[scale=1.0]}] plot [rounded corners=0.25cm] coordinates { (-0.8,-5)(-0.8,-5.5)
};
\draw [draw=red,draw opacity=0.5,line width=4pt] plot [rounded corners=0.25cm] coordinates { (0.4,-4)(-0.8,-4)(-0.8,-6)(-0.4,-5.97)
};
\node[color=black] (a) at (4.8, -3.65) {towards $Z_e$};
\node[color=red] (a) at (-0.4, -4.45) {$\mathcal{Z}_0^{(2)}$};
\end{tikzpicture}
\caption{Odd $N$ and $\mathcal{Z}_0^{(2)}$.}
\end{subfigure}\\[10pt]
\begin{subfigure}[b]{\linewidth}
\centering%
\begin{tikzpicture}
%
\draw [black,-{Latex[scale=1.0]}] plot [rounded corners=0.25cm] coordinates { (8.9,2)(6.7,2)(6.7,0)(13.9,0)(13.9,2)(7.95,2)
(7.5,1.9)(7.5,0.1)(8.3,0.1)(8.3,1.9)(9.1,1.9)(9.1,0.1)(9.9,0.1)(9.9,1.9)(10.7,1.9)(10.7,0.1)(11.5,0.1)(11.5,1.9)(12.3,1.9)(12.3,0.1)(13.1,0.1)(13.1,1.9)(12.7,1.9)
};
\draw [black,-{Latex[scale=1.0]}] plot [rounded corners=0.25cm] coordinates { (13.1,1)(13.1,1.9)(13.5,1.9)
}; 
\draw [black,-{Latex[scale=1.0]}] plot [rounded corners=0.25cm] coordinates { (6.7,1.5)(6.7,0.5)
};
\draw[dotted,blue,line width=1pt] (10.3, -0.2) -- (10.3, 2.2);
\draw[solid,draw=red,draw opacity=0.5,line width=4pt] (8,2) -- (10.3, 2);
\node[color=red] (a) at (9.1, 2.4) {$\mathcal{Z}_0^{(1)}$};
\node[color=black] (a) at (13.1, 2.35) {towards $Z_e$};
\end{tikzpicture}
\caption{Even $N$ and $\mathcal{Z}_0^{(1)}$.}
\end{subfigure}\\[10pt]
\begin{subfigure}[b]{\linewidth}
\centering%
\begin{tikzpicture}
\draw [black,-{Latex[scale=1.0]}] plot [rounded corners=0.25cm] coordinates { (6.7,-5.75)(6.7,-6)(13.9,-6)(13.9,-4)(6.7,-4)(6.7,-6)(7.5,-5.95)(7.5,-4.1)(8.3,-4.1)(8.3,-5.9)(9.1,-5.9)(9.1,-4.1)(9.9,-4.1)(9.9,-5.9)(10.7,-5.9)(10.7,-4.1)(11.5,-4.1)(11.5,-5.9)(12.3,-5.9)(12.3,-4.1)(13.1,-4.1)(13.1,-5.9)(12.7,-5.9)
};
\draw [black,-{Latex[scale=1.0]}] plot [rounded corners=0.25cm] coordinates { (13.1,-5)(13.1,-5.9)(13.5,-5.9)
}; 
\draw[dotted,blue,line width=1pt] (10.3, -6.2) -- (10.3, -3.8);
\draw [black,-{Latex[scale=1.0]}] plot [rounded corners=0.25cm] coordinates { (6.7,-5)(6.7,-5.5)
};
\draw [draw=red,draw opacity=0.5,line width=4pt] plot [rounded corners=0.25cm] coordinates { (7.9,-4)(6.7,-4)(6.7,-6)(7.1,-5.97)
};
\node[color=black] (a) at (13.1, -6.35) {towards $Z_e$};
\node[color=red] (a) at (7.1, -4.45) {$\mathcal{Z}_0^{(2)}$};
\end{tikzpicture}
\caption{Even $N$ and $\mathcal{Z}_0^{(2)}$.}
\end{subfigure}
\caption{ABp. Distinction between four combinations of odd and even $N$ and, in red, two possible sets of start (and simultaneously end) positions, $\mathcal{Z}_0^{(1)}$ and $\mathcal{Z}_0^{(2)}$. The blue dotted lines indicate the $\frac{\xi_1 + \xi_N}{2}$-coordinate, respectively. Note that only for better visualisation of the route planning logic, the meandering path is not displayed as coinciding with the headland path.}
\label{fig_MEANDER}
\end{figure}
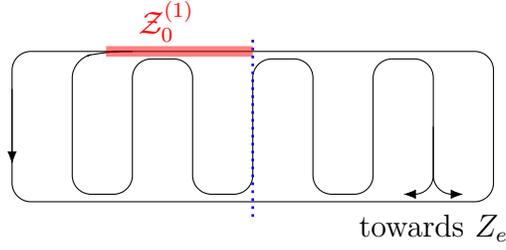
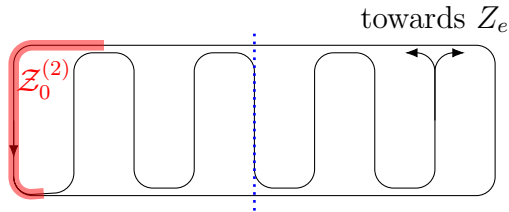
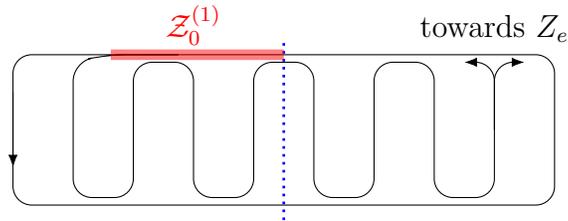
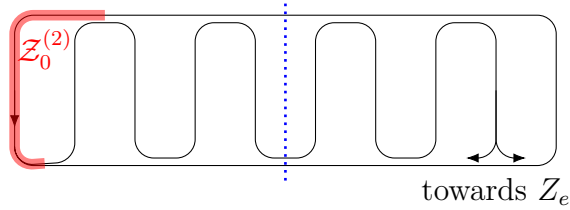


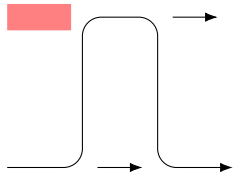
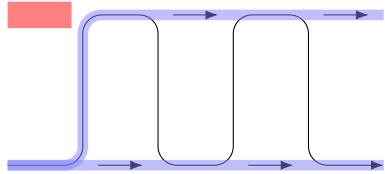
\begin{figure}
\vspace{0.3cm}
\begin{subfigure}[b]{\linewidth}
\centering%
\begin{tikzpicture}
\draw [black,-{Latex[scale=1.0]}] plot [rounded corners=0.25cm] coordinates { (0,0)(1,0)(1,2)(2,2)(2,0)(3,0)
};
\draw [black,-{Latex[scale=1.0]}] plot [rounded corners=0.25cm] coordinates { (1.2,0)(1.8,0)};
\draw [black,-{Latex[scale=1.0]}] plot [rounded corners=0.25cm] coordinates { (2.2,2)(2.8,2)};
\draw [draw=red,draw opacity=0.5, line width=10pt] plot [rounded corners=0.25cm] coordinates { (0.85,2)(0,2)};
\end{tikzpicture}
\caption{The path planning pattern for ABp.}
\end{subfigure}\\[10pt]%
\begin{subfigure}[b]{\linewidth}
\centering%
\begin{tikzpicture}
\draw [black,-{Latex[scale=1.0]}] plot [rounded corners=0.25cm] coordinates { (0,-3.5)(1,-3.5)(1,-1.5)(2,-1.5)(2,-3.5)(3,-3.5)(3,-1.5)(4,-1.5)(4,-3.5)(5,-3.5)
};
\draw [black,-{Latex[scale=1.0]}] plot [rounded corners=0.25cm] coordinates { (1.2,-3.5)(1.8,-3.5)};
\draw [black,-{Latex[scale=1.0]}] plot [rounded corners=0.25cm] coordinates { (2.2,-1.5)(2.8,-1.5)};
\draw [black,-{Latex[scale=1.0]}] plot [rounded corners=0.25cm] coordinates { (3.2,-3.5)(3.8,-3.5)};
\draw [black,-{Latex[scale=1.0]}] plot [rounded corners=0.25cm] coordinates { (4.2,-1.5)(4.8,-1.5)};
\draw [draw=red,draw opacity=0.5, line width=10pt] plot [rounded corners=0.25cm] coordinates { (0.85,-1.5)(0,-1.5)};
\draw [draw=blue!50,draw opacity=0.5, line width=4pt] plot [rounded corners=0.25cm] coordinates { (0,-3.5)(1,-3.5)(1,-1.5)(2,-1.5)(5,-1.5)};
\draw [draw=blue!50,draw opacity=0.5, line width=4pt] plot [rounded corners=0.25cm] coordinates { (0,-3.5)(5,-3.5)};
\end{tikzpicture}
\caption{Concatenation of two path planning patterns for ABp.}
\end{subfigure}
\caption{ABp. (Top plot) Illustration of the path planning pattern. The red bar indicates the area that cannot be reached by neither traversal of the path planning pattern nor traversal of the headland segments in the directions as indicated by the arrows. (Bottom plot) Concatenation of two pattern elements. The traversal along the ``upper'' and ``lower'' headland path is emphasised in blue. Importantly, the area indicated by the red bar can still \emph{not} be reached, see Proposition \ref{Prop_MEAN_HaveToReachN}.}
\label{fig_MEANDER_1unit}
\end{figure}

In current working practice, the overwhelming majority of field coverage paths is planned based on sequential concatenation of geometrically translated \emph{AB lines} such that a meandering path is generated, see \cite{palmer2003improving} and Fig. \ref{fig_MEANDER}. In the following, this method of field coverage is referred to as ABp, which is abbreviated for \emph{AB pattern} (terminology from \cite{bochtis2013benefits}). The fact that it is so widespread can be observed from satellite images and their display of tyre traces. Field coverage can be decomposed into, first, the traveling along the headland path, and, second, the subsequent following of the meandering path until completion of field coverage. Under the assumptions of i) a continuous and thorough initial headland path traversal before interior lane coverage, ii) field shapes according to the description of Section \ref{sec_problFormulation}, and iii) field coverage in \emph{one} single traversal without requiring an intermediate return to a stationary or mobile depot for refilling (or a similar task), the path planning method based on the AB pattern is the optimal strategy, even under compacted area minimisation constraints. This is since the non-working distance is minimised. It is constrained to the headland path segments that were already covered during the initial headland path traversal. However, as will be shown, this method is not optimal for \emph{partial} field coverage. Because of its widespread usage it will serve as the baseline. In the remainder of this section, its optimal demployment is discussed and its intrinsic disadvantages for partial field coverage are exposed. 

Four cases of combinations of the number of interior lanes $N$ and the set of start positions $\mathcal{Z}_0$ are distinguished. They are:~(odd $N$, $\mathcal{Z}_0^{(1)}$), (even $N$, $\mathcal{Z}_0^{(1)}$),  (odd $N$, $\mathcal{Z}_0^{(2)}$) and (even $N$, $\mathcal{Z}_0^{(2)}$). These four entail path planning as displayed in Fig. \ref{fig_MEANDER} for 7 lanes (unven $N$) and 8 lanes (odd $N$). The mathematical description of $\mathcal{Z}_0^{(1)}$ and $\mathcal{Z}_0^{(2)}$ is derived as follows. Based on the definitions in Fig. \ref{fig_IllustrNotation}, $\mathcal{H}=\{(\xi,\eta): (\xi,\eta)\in\text{headland path} \}$ is defined, i.e., as the set of $(\xi,\eta)$-coordinates along the headland path. The auxiliary location $Z_M = (\xi_M,\eta_M)$ with $\xi_M = (\xi_1 + \xi_{N})/2$ and $\eta_M = \max~\{\eta: (\xi,\eta)\in\mathcal{H},~\xi=\xi_M\}$ are further defined. Then, the path length coordinate $s$ is initialised along the headland path at $Z_M$ with $s_M=0$. This permits to define two sets of path coordinates for field entrance positions, i.e., $\mathcal{S}_0^{(1)} = \{s: 0\leq s \leq s_{N+1} \}$ and  $\mathcal{S}_0^{(2)} = \{s: s_{N+1} \leq s \leq s_{1} \}$, whereby $s_{N+1}$ and $s_1$ denot the path length coordinates at location $Z_{N+1}$ and $Z_1$, respectively. Consequently, the two sets of possible start positions expressed within the normalised coordinate framework can be defined as $\mathcal{Z}_0^{(l)} = \{ Z(s): s\in\mathcal{S}_0^{(l)}\}$ for $l=1,2$, and where $Z(s)$ denotes a location at path length coordinate $s$ along the headland path. 
  
Let us elaborate on the employed coordinate system transformations. The rotation step transforms coordinates from the global $(x,y)$-description to the $(x_\theta,y_\theta)$-coordinate system such that interior lanes are aligned to the $y_\theta$-axis. Then, at most two additional coordinate reflections are applied. Therefore, first $x_{\theta,M} = (\max_{x_\theta\in\mathcal{X}_\theta} x_\theta - \min_{x_\theta\in\mathcal{X}_\theta} x_\theta)/2$ is defined with $\mathcal{X}_\theta$ denoting the set of all $x_\theta$-coordinates defining the headland path, before applying the first coordinate reflection by the linear mapping:

\begin{align}
x_\theta^x &= x_{\theta,M} - (x_\theta-x_{\theta,M}),\label{eq_def_mirror1_x}\\
y_\theta^x &= y_\theta.\label{eq_def_mirror1_y}
\end{align}

If the transformed coordinates are not yet sufficiently normalised such that the starting position falls into above framework and according to Fig. \ref{fig_MEANDER}, then the following second coordinate reflection is applied:

\begin{align}
x_\theta^{xy} &= x_\theta^{x},\label{eq_def_mirror2_x}\\
y_\theta^{xy} &= y_{\theta,M}^x - (y_{\theta}^x - y_{\theta,M}^x),\label{eq_def_mirror2_y}
\end{align}
with $y_{\theta,M}^x =  (\max_{y_\theta^x\in\mathcal{Y}_\theta^x} y_{\theta}^x - \min_{y_\theta^x\in\mathcal{Y}_\theta^x} y_{\theta}^x)/2$, and where $\mathcal{Y}_\theta^x$ describes the set of $y_{\theta}^x$-coordinates defining the headland path. At the latest after this second transformation, coordinates are normalised such that the starting position falls into above framework, see Fig. \ref{fig_MEANDER}. Thus, $(\xi,\eta)$ represents the $(x_\theta,y_\theta)$-, $(x_\theta^x,y_\theta^x)$-, or $(x_\theta^{xy},y_\theta^{xy})$-coordinate system. After path planning in the normalised coordinate system, all linear mappings required for normalisation must be inverted to obtain the result within the $(x,y)$-plane. 

The path planning method for ABp is summarised in the following Algorithm. It describes the offline fitting of a traversable path to a given field.\\[3pt]

\begin{tabular}{ l l }
\hline
\multicolumn{2}{l}{Algorithm 1: ABp (offline)}\\[2pt]
 1. & Normalisation of the coordinate system description:\\
 & - one rotation, and at most two reflection steps.\\
  & - description in the $(\xi,\eta)$-plane.\\[2pt]
 2. & Distinction between four cases:\\
 & - four combinations of even/odd $N$ and $\mathcal{Z}_0^{(1)}$/$\mathcal{Z}_0^{(2)}$.\\[2pt]  
 3. & Path planning according to the method of Fig. \ref{fig_MEANDER}. \\[2pt]
 4. & Retransformation of coordinates to the $(x,y)$-plane.\\[2pt]
 \hline \\[2pt]
\label{alg_meand_offline}  
\end{tabular}

Characteristics of the method are discussed. The pattern on which ABp is founded is displayed in Fig. \ref{fig_MEANDER_1unit}. ``Lower'' and ``upper'' headland segments are defined as the set of edges $\{e_{i,j}: i,j=1,\dots,N,2N+1\}$ and $\{e_{i,j}: i,j=N+1,\dots,2N,0,2N+2\}$, respectively. It is distinguished between two possible methods for the transition from lane $N$ towards the headland path; see the labeling ``towards $Z_e$'' in Fig. \ref{fig_MEANDER}. For all four cases, (a)--(d), the direction pointing towards $\mathcal{Z}_0$, which is also the method implicitly assumed for the remainder of Section \ref{subsec_MEAND}, is more favourable than its alternative with regard of path length minimisation from any given position back towards $Z_0$. This is easy to see from the fact that a transition from interior lane to headland path is created, which must be respected as a compacted area minimisation constraint. Even if $Z_e$ is located such that $\xi_{e}>\xi_{N}$, for overall path length minimisation, the method with a final transition pointing towards smaller $\xi_{0}$ may typically still be preferable. This holds especially when frequent returns to a mobile depot are required. Importantly, it also guarantees that after traversal of the $N$th lane, a move along the ``upper'' headland path heading towards $Z_0$ is then possible.

\begin{proposition}\label{Prop_MEAN_HaveToReachN}
Assume a normalised coordinate system description with $Z_0=(\xi_0,\eta_0)$ according to Fig. \ref{fig_MEANDER}, in which it is accounted for compacted area minimisation constraints, and in which it is aimed at finding the shortest path from position $Z(t)=(\xi(t),\eta(t))$ at time $t$ to $Z_0$. Then, if $\xi(t)\geq \xi_{0}$, a corresponding agricultural vehicle has to \emph{always} traverse the ultimate lane $N$ as part of the path to reach $Z_0$, unless it already has covered all interior lanes and is heading back towards $Z_0$ along the ``upper'' headland path, or unless it is heading towards $Z_0$ along the ``upper'' headland path as part of the initial headland path traversal.
\end{proposition}
\begin{proof}
The proof is by construction and follows directly from the meandering path motif in Fig. \ref{fig_MEANDER_1unit}, and the assumption of complying with compacted area minimisation constraints. See Fig. \ref{fig_MEANDER} for visualisation.
\end{proof}

Proposition \ref{Prop_MEAN_HaveToReachN} is particularly relevant for fields with \emph{many} lanes $N$ (``fat'' fields). For $\xi(t)<\xi_0$, no such generalizing statement can be made without making further differentiations between even/odd $N$ and $\mathcal{Z}_0^{(l)}$ for $l=1,2$. However, Proposition \ref{Prop_MEAN_HaveToReachN} can be generalised to alternative locations different from $Z_0$. See Fig. \ref{fig_MEANDER_1unit} for visualisation of areas that cannot be reached without reaching the final lane after the concatenation of multiple pattern elements.

\begin{remark}
\label{rmk_MEAND_Astar_ZtZ0}
While Proposition \ref{Prop_MEAN_HaveToReachN} guarantees that the last lane $N$ must be reached, no generalizing statement can be made with respect to the shortest path for reaching it. The $\xi$-coordinate must be monotonically increasing throughout the process of reaching it. However, it does  \emph{not} necessarily have to be \emph{strictly} monotonically increasing. For example, depending on the field contour and orientation of interior lanes, the shortest path may involve transitions along interior lanes from ``upper'' to ``lower'' headland segments, vice versa, and even multiple times during the process of reaching lane $N$. After traversal of the $N$th lane, the shortest path to $Z_0$ is sought. Here, the same concepts apply. Strictly monotonous movement towards $Z_0$ is not required. Potential transitions between ``upper'' and ``lower'' headland path segments may be path length optimal. In practice, a shortest path algorithm, see \cite{bertsekas1995dynamic}, can be employed for both i) the reaching of the $N$th lane starting from location $Z(t)$ at time $t$, and ii) the reaching of $Z_0$ after the traversal of the $N$th lane.
\end{remark}

\begin{remark}
\label{rmk_MEAND_resumption}
For the \emph{resumption} of work at location $Z(\tau^\text{last}) = (\xi(\tau^\text{last}),\eta(\tau^\text{last}))$, the shortest possible path from $Z_0$ to that location can be selected if $\xi(\tau^\text{last})>\xi_0$ and $Z(\tau^\text{last})$ is not located along the ``upper'' headland. This can be seen from Fig. \ref{fig_MEANDER_1unit}; there is no transition from headland to interior lane or vice versa that is prohibiting such shortest path. In contrast, for alternative locations of $Z(\tau^\text{last})$, the heading direction along the field coverage path plays an important role. To resume a specific heading orientation, a path may have to be taken that is deviating from the shortest path connecting $Z_0$ and $Z(\tau^\text{last})$.
\end{remark}

At any time $t$, the agricultural vehicle can be in any of three states $\gamma(t)\in\{0,1,2\}$. The case $\gamma(t)=0$ corresponds to a mode in which the vehicle is on its way back to the \emph{resuming} location $Z(\tau^\text{last})$ at which the field coverage was terminated last at time $\tau^\text{last}$. The case $\gamma(t)=1$ corresponds to a mode in which the vehicle is following the field coverage path plan according to Fig. \ref{fig_MEANDER}. Finally, $\gamma(t)=2$ indicates the mode in which the vehicle is in the process of \emph{returning} to $Z_0$ for refilling. During operation the vehicle alternates between any of these three states dependent on its storage tank fill-level.

The predicted fill-level at time $t+\Delta t$ is denoted by $\hat{f}(t+\Delta t)$, and bounded between $0$ and the maximum fill capacity, and $\Delta t$ the time discretisation. During online operation and dependent on the location of $Z_0$, it often is favourable to trigger a return command on the last lane with heading in negative $\eta$-direction before the fill-level is about to reach zero. This is since the resulting path to $Z_0$ typically involves the following steps: i) completion of the current interior lane,  ii) traveling along the ``lower'' headland path, iii) a transition to the ``upper'' headland path via interior lane $N$, and iv) traveling along the ``upper'' headland path until $Z_{0}$. The examples of Section \ref{sec_IllustrativeEx} will further illustrate this consideration. In general, the decision upon when to trigger the return command must trade-off current fill-level $f(t)\geq 0$, the shortest path length $P(Z(t),Z_0)$ from the current location $Z(t)$ to the depot $Z_0$, and the path length $P(\hat{Z}_{\hat{f}=0},Z_0)$ from the predicted location $\hat{Z}_{\hat{f}=0}$ at which the fill-level is expected to reach zero to the depot $Z_0$. A return command may be triggered if $ f(t)$ is smaller than a small threshold fill-level (above which a return to $Z_0$ is undesired) and $P(Z(t),Z_0)<P(\hat{Z}_{\hat{f}=0},Z_0)$. Alternatively, it must be triggered if $f(t)=0$. For the prediction of $\hat{f}(t)$ and $\hat{Z}_{\hat{f}=0}$, experience from past spraying maps may be used. Alternatively, a linear parameter varying model, $f(t+\Delta t)=f(t)-a_f(t)\Delta t$ can be assumed, where $a_f(t)$ denotes the time-varying emptying rate parameter (for variable rate spray applications). Then, at every $t$, a \emph{discrete Extended Kalman Filter} (\cite{anderson1979optimal}) can be employed to provide estimates $\hat{f}(t)$ and $\hat{a}_f(t)$, based on which $\hat{Z}_{\hat{f}=0}$ can be predicted through model forward simulation. Note that this approach enables sensor fusion of possibly multiple measurements (e.g., fill level and nozzles).

\subsection{Path Planning based on the Circular Pattern -- CIRC\label{subsec_CIRC}}

The path planning method for partial field coverage labeled CIRC is visualised in Fig. \ref{fig_CIRCULAR}. It is referred to as CIRC because of its circular path planning pattern, see Fig. \ref{fig_CIRCULAR_1unit}. The method is summarised in the following Algorithm:\\[3pt]  

\begin{tabular}{ l l }
\hline
 \multicolumn{2}{l}{Algorithm 3: CIRC (offline)}\\[2pt]
 1. & Normalisation of the coordinate system description:\\
 & - one rotation, at most two reflection steps.\\
  & - description in the $(\xi,\eta)$-plane.\\[2pt]
 2. & Distinction between four cases:\\
 & - four combinations of even/odd $N$ and $\mathcal{Z}_0^{(1)}$/$\mathcal{Z}_0^{(2)}$.\\[2pt]  
 3. & Path planning according to the method of Fig. \ref{fig_CIRCULAR}. \\[2pt]
 4. & Retransformation of coordinates to the $(x,y)$-plane.\\[2pt]
 \hline \\[2pt]  
\end{tabular}

\begin{figure}
\begin{subfigure}[b]{\linewidth}
\centering%
\begin{tikzpicture}
\draw [black,-{Latex[scale=1.0]}] plot [rounded corners=0.25cm] coordinates { (0.4,2)(-0.8,2)(-0.8,0)(5.6,0)(5.6,2)(0.45,2)
(0,1.9)(0,0.1)(1.6,0.1)(1.6,1.9)(0.8,1.9)(0.8,0.2)(3.2,0.1)(3.2,1.9)(2.4,1.9)(2.4,0.2)(4.8,0.1)(4.8,1.9)(4.0,1.9)(4.0,0.2)(4.8,0.16)(4.8,0.85)
};
\draw[dotted,blue,line width=1pt] (2.4, -0.2) -- (2.4, 2.2);
\draw [black,-{Latex[scale=1.0]}] plot [rounded corners=0.25cm] coordinates { (-0.8,1.5)(-0.8,1)
};
\draw [black,-{Latex[scale=1.0]}] plot [rounded corners=0.25cm] coordinates { (4.5,0.17)(5.2,0.17)};
%
\draw [draw=red,draw opacity=0.5, line width=4pt] plot [rounded corners=0.25cm] coordinates { (2.4,2)(0.45,2)
};
\node[color=red] (a) at (1.25, 2.4) {$\mathcal{Z}_0^{(1)}$};
\node[color=black] (a) at (4.8, -0.35) {towards $Z_e$};
\end{tikzpicture}
\caption{Odd $N$ and $\mathcal{Z}_0^{(1)}$.}
\end{subfigure}\\[0pt]
\begin{subfigure}[b]{\linewidth}
\centering%
\begin{tikzpicture}
\draw [black,-{Latex[scale=1.0]}] plot [rounded corners=0.25cm] coordinates { (-0.8,-4.75)(-0.8,-6)(5.6,-6)(5.6,-4)(-0.8,-4)(-0.8,-5.9)(0.8,-5.9)(0.8,-4.1)(0,-4.1)(0,-5.8)(2.4,-5.9)(2.4,-4.1)(1.6,-4.1)(1.6,-5.8)(4.0,-5.9)(4.0,-4.1)(3.2,-4.1)(3.2,-5.8)(4.8,-5.87)(4.8,-4.1)(4.3,-4.1)
};
\draw [black,-{Latex[scale=1.0]}] plot [rounded corners=0.25cm] coordinates { (4.8,-5)(4.8,-4.1)(5.2,-4.1)
}; 
\draw[dotted,blue,line width=1pt] (2.4, -6.2) -- (2.4, -3.8);
\draw [black,-{Latex[scale=1.0]}] plot [rounded corners=0.25cm] coordinates { (-0.8,-4.5)(-0.8,-5)
};
\draw [draw=red,draw opacity=0.5,line width=4pt] plot [rounded corners=0.25cm] coordinates { (0.45,-4)(-0.8,-4)(-0.8,-6)(0.35,-5.97)
};
\node[color=black] (a) at (4.8, -3.65) {towards $Z_e$};
\node[color=red] (a) at (-0.4, -4.45) {$\mathcal{Z}_0^{(2)}$};
\end{tikzpicture}
\caption{Odd $N$ and $\mathcal{Z}_0^{(2)}$.}
\end{subfigure}\\[0pt]
\begin{subfigure}[b]{\linewidth}
\centering%
\begin{tikzpicture}
\draw [black,-{Latex[scale=1.0]}] plot [rounded corners=0.25cm] coordinates { (8.9,2)(6.7,2)(6.7,0)(13.9,0)(13.9,2)(7.95,2)
(7.5,1.9)(7.5,0.2)(9.1,0.1)(9.1,1.9)(8.3,1.9)(8.3,0.2)(10.7,0.1)(10.7,1.9)(9.9,1.9)(9.9,0.2)(12.3,0.1)(12.3,1.9)(11.5,1.9)(11.5,0.2)(13.1,0.1)(13.1,1.9)(12.7,1.9)
};
\draw [black,-{Latex[scale=1.0]}] plot [rounded corners=0.25cm] coordinates { (13.1,1)(13.1,1.9)(13.5,1.9)
}; 
\draw [black,-{Latex[scale=1.0]}] plot [rounded corners=0.25cm] coordinates { (6.7,1.5)(6.7,1)
};
\draw[dotted,blue,line width=1pt] (10.3, -0.2) -- (10.3, 2.2);
\draw[solid,draw=red,draw opacity=0.5,line width=4pt] (8,2) -- (10.3, 2);
\node[color=red] (a) at (9.1, 2.4) {$\mathcal{Z}_0^{(1)}$};
\node[color=black] (a) at (13.1, 2.35) {towards $Z_e$};
\end{tikzpicture}
\caption{Even $N$ and $\mathcal{Z}_0^{(1)}$.}
\end{subfigure}\\[0pt]
\begin{subfigure}[b]{\linewidth}
\centering%
\begin{tikzpicture}
\draw [black,-{Latex[scale=1.0]}] plot [rounded corners=0.25cm] coordinates { (6.7,-4.75)(6.7,-6)(13.9,-6)(13.9,-4)(6.7,-4)(6.7,-5.9)(8.3,-5.9)(8.3,-4.1)(7.5,-4.1)(7.5,-5.8)(9.9,-5.9)(9.9,-4.1)(9.1,-4.1)(9.1,-5.8)(11.5,-5.9)(11.5,-4.1)(10.7,-4.1)(10.7,-5.8)(13.1,-5.9)(13.1,-4.1)(12.3,-4.1)(12.3,-5.8)(13.1,-5.84)(13.1,-5.15)
};
\draw [black,-{Latex[scale=1.0]}] plot [rounded corners=0.25cm] coordinates { (12.8,-5.83)(13.5,-5.83)
}; 
\draw[dotted,blue,line width=1pt] (10.3, -6.2) -- (10.3, -3.8);
\draw [black,-{Latex[scale=1.0]}] plot [rounded corners=0.25cm] coordinates { (6.7,-4.5)(6.7,-5)
};
\draw [draw=red,draw opacity=0.5,line width=4pt] plot [rounded corners=0.25cm] coordinates { (7.95,-4)(6.7,-4)(6.7,-6)(7.85,-5.97)
};
\node[color=black] (a) at (13.1, -6.35) {towards $Z_e$};
\node[color=red] (a) at (7.1, -4.45) {$\mathcal{Z}_0^{(2)}$};
\end{tikzpicture}
\caption{Even $N$ and $\mathcal{Z}_0^{(2)}$.}
\end{subfigure}
\caption{CIRC. Distinction between four cases of combinations of odd and even $N$, and in red, two possible sets of start positions, $\mathcal{Z}_0^{(1)}$ and $\mathcal{Z}_0^{(2)}$. The blue dotted lines indicate the $\frac{\xi_1 + \xi_N}{2}$-coordinate, respectively.}
\label{fig_CIRCULAR}
\end{figure}
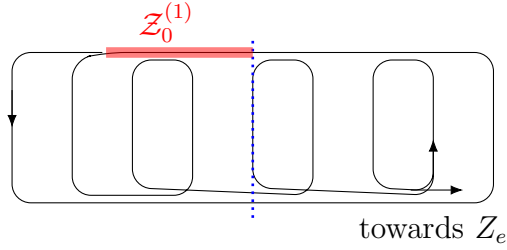
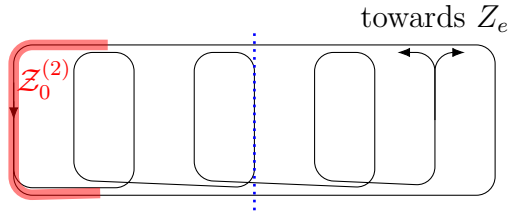
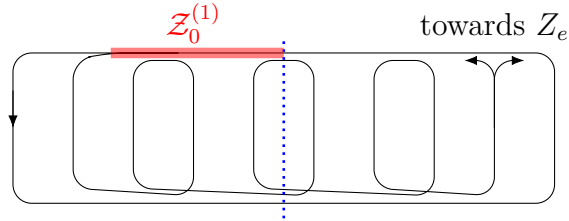
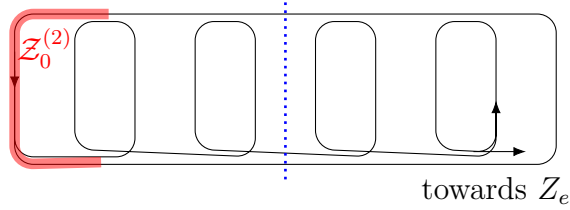

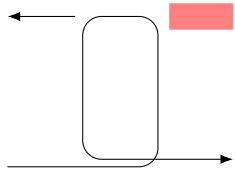
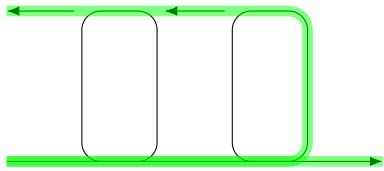
\begin{figure}
\vspace{0.3cm}
\begin{subfigure}[b]{\linewidth}
\centering%
\begin{tikzpicture}
\draw [black,-{Latex[scale=1.0]}] plot [rounded corners=0.25cm] coordinates { (0,0)(2,0)(2,2)(1,2)(1,0.1)(3,0.1)
};
\draw [black,-{Latex[scale=1.0]}] plot [rounded corners=0.25cm] coordinates { (0.9,2)(0,2)};
%
\draw [draw=red,draw opacity=0.5, line width=10pt] plot [rounded corners=0.25cm] coordinates { (2.15,2)(3,2)};
\end{tikzpicture}
\caption{Path planning pattern for CIRC and CIRC$^\star$.}
\end{subfigure}\\[10pt]%
\begin{subfigure}[b]{\linewidth}
\centering%
\begin{tikzpicture}
\draw [black,-{Latex[scale=1.0]}] plot [rounded corners=0.25cm] coordinates { (0,-3.5)(2,-3.5)(2,-1.5)(1,-1.5)(1,-3.5)(4,-3.5)(4,-1.5)(3,-1.5)(3,-3.5)(5,-3.5)
};
\draw [black,-{Latex[scale=1.0]}] plot [rounded corners=0.25cm] coordinates { (0.9,-1.5)(0,-1.5)};
\draw [black,-{Latex[scale=1.0]}] plot [rounded corners=0.25cm] coordinates { (2.9,-1.5)(2.1,-1.5)};
\draw [draw=green,draw opacity=0.5, line width=4pt] plot [rounded corners=0.25cm] coordinates { (0,-3.5)(4,-3.5)(4,-1.5)(0,-1.5)};
\draw [draw=green,draw opacity=0.5, line width=4pt] plot [rounded corners=0.25cm] coordinates { (0,-3.5)(5,-3.5)};
\end{tikzpicture}
\caption{Concatenation of two path planning patterns.}
\end{subfigure}
\caption{
CIRC and CIRC$^\star$. (Top plot) Illustration of a path planning pattern element. The red bar indicates the area that cannot be reached by neither traversal of the path planning pattern element nor traversal of the headland segments in the directions of the arrows. (Bottom plot) Concatenation of two patterns. Importantly, the area indicated by the red bar in the top plot can now be reached. See the green paths for emphasis of the traversal along the ``lower'' headland path, and a transition via an interior lane to the ``upper'' headland path. 
}
\label{fig_CIRCULAR_1unit}
\end{figure}

Note that, while not in the \emph{partial} field coverage context, the same pattern is also employed in \cite{bochtis2013benefits}. There, it is referred to as ``Skip and Fill''-pattern. Here, the label ``CIRC'' is preferred to avoid confusion with the fill-level of storage tanks.


\begin{proposition}\label{Prop_CIRC_ZtToZ0}
Assume a normalised coordinate system description with $Z_0=(\xi_0,\eta_0)$ according to Fig. \ref{fig_CIRCULAR}, in which it is accounted for compacted area minimisation constraints, and in which the goal is to find the shortest path from position $Z(t)=(\xi(t),\eta(t))$ at time $t$ to $Z_0$. Then, for $\xi_1 \leq \xi(t)< \xi_{N-1}$, an agricultural vehicle can always traverse at the latest the second next interior lane such that afterwards it can travel along the ``upper'' headland path in direction of $Z_0$. For $\xi_1\leq \xi(t) < \xi_{N-1}$, either lane $n$ or $n+1$ permit such traversal, whereby $n$ is such that $\xi_{n-1}\leq \xi(t)<\xi_n$. For $\xi(t)<\xi_1$, in general, either lane $n=2$ or $n=3$ permit such traversal, see Fig. \ref{fig_CIRCULAR}. For $\xi(t)\geq\xi_{N-1}$, lane $n=N$ or a path through edges $e_{N,2N+2}$ and $e_{2N+2,2N}$ permits the traversal. 
\end{proposition}
\begin{proof}
The proof is by construction and follows directly from the circular path motif in Fig. \ref{fig_CIRCULAR_1unit}, and the assumption of complying with compacted area minimisation constraints. See also Fig. \ref{fig_CIRCULAR} for visualisation. 
\end{proof}

\begin{remark}
\label{rmk_CIRC_EnforcedTravelUpperHeadl}
An implication of CIRC is that once the ``upper'' headland path is reached (assuming a normalised coordinate system description), the vehicle is constrained to travel along it until reaching $Z_0$. This is because of the compacted area minimisation constraint and the characteristic pattern according to Fig. \ref{fig_CIRCULAR_1unit}. This is also in contrast to the ABp-method and implies that \emph{no} additional invoking of a shortest path algorithm is required. While for ABp the shortest path to $Z_0$ after reaching of the $N$th lane may, in general, involve traversals of interior lanes and thus switching between ``upper'' and ``lower'' headland paths, this is not the case for CIRC. This distinction is the reason that no guarantee can be given about a shorter path length for CIRC. Consider an extremely large bulb-like headland segment located between two interior lanes. While ABp can avoid this by a traversal to the ``lower'' headland path, the method according to CIRC is \emph{enforced} to traverse it. Note that such (theoretical) scenarios are seldom in practice. In Section \ref{sec_IllustrativeEx}, a quantitative comparison for a rectangular field is given as a function of $N$, the length of interior lanes and the machine operating width.
\end{remark}


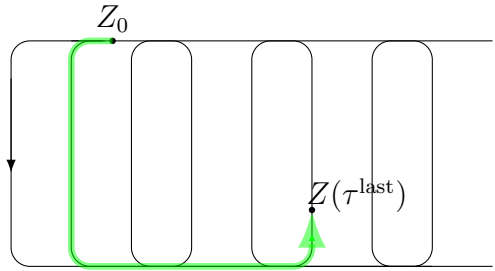
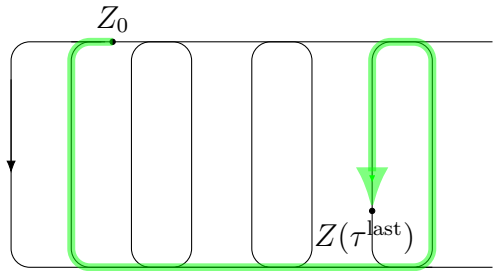
\begin{figure}
\begin{subfigure}[t]{\linewidth}
\centering%
\begin{tikzpicture}
\draw [black] plot [rounded corners=0.25cm] coordinates { (0.4,-2.25)(-0.8,-2.25)(-0.8,-5.25)(5.6,-5.25)
};
\draw [black] plot [rounded corners=0.25cm] coordinates { (0,-2.25)(5.6,-2.25)
};
\draw [black] plot [rounded corners=0.25cm] coordinates {
(0.45,-2.25)(0,-2.25)
(0,-5.25)(1.6,-5.25)(1.6,-2.25)(0.8,-2.25)(0.8,-5.25)
(1.6,-5.25)(3.2,-5.25)(3.2,-2.25)(2.4,-2.25)(2.4,-5.25)
(3.2,-5.25)(4.8,-5.25)(4.8,-2.25)(4,-2.25)(4,-5.25)
(4.8,-5.25)
};
\draw[fill=black] (0.55,-2.25) circle (1pt); 
\node[color=black] (a) at (0.55, -1.95) {$Z_0$};
%
\draw [black,-{Latex[scale=1.0]}] plot [rounded corners=0.25cm] coordinates { (-0.8,-2.75)(-0.8,-4)
};
\draw[fill=black] (3.2,-4.5) circle (1pt); 
\node[color=black] (a) at (3.8, -4.3) {$Z(\tau^\text{last})$};
\draw [draw=green,draw opacity=0.5, line width=2.5pt,-{Latex[scale=1.0]}] plot [rounded corners=0.25cm] coordinates { (0.55,-2.25)(0,-2.25)(0,-5.25)(3.2,-5.25)(3.2,-4.5)};
\end{tikzpicture}
\caption{First scenario.}
\end{subfigure}\\[0pt]%
\begin{subfigure}[b]{\linewidth}
\centering%
\begin{tikzpicture}
\draw [black] plot [rounded corners=0.25cm] coordinates { (0.4,-2.25)(-0.8,-2.25)(-0.8,-5.25)(5.6,-5.25)
};
\draw [black] plot [rounded corners=0.25cm] coordinates { (0,-2.25)(5.6,-2.25)
};
\draw [black] plot [rounded corners=0.25cm] coordinates {
(0.45,-2.25)(0,-2.25)
(0,-5.25)(1.6,-5.25)(1.6,-2.25)(0.8,-2.25)(0.8,-5.25)
(1.6,-5.25)(3.2,-5.25)(3.2,-2.25)(2.4,-2.25)(2.4,-5.25)
(3.2,-5.25)(4.8,-5.25)(4.8,-2.25)(4,-2.25)(4,-5.25)
(4.8,-5.25)
};
\draw[fill=black] (0.55,-2.25) circle (1pt); 
\node[color=black] (a) at (0.55, -1.95) {$Z_0$};
%
\draw [black,-{Latex[scale=1.0]}] plot [rounded corners=0.25cm] coordinates { (-0.8,-2.75)(-0.8,-4)
};
\draw[fill=black] (4,-4.5) circle (1pt); 
\node[color=black] (a) at (3.9, -4.85) {$Z(\tau^\text{last})$};
\draw [draw=green,draw opacity=0.5, line width=3pt,-{Latex[scale=1.0]}] plot [rounded corners=0.25cm] coordinates { (0.55,-2.25)(0,-2.25)(0,-5.25)(4.8,-5.25)(4.8,-2.25)(4,-2.25)(4,-4.5)};
\end{tikzpicture}
\caption{Second scenario.}
\end{subfigure}
\caption{CIRC. (Top plot) Resuming work at a position $Z(\tau^\text{last})$ located along a lane with heading direction towards positive $\eta$. (Bottom plot) Resuming work at a position $Z(\tau^\text{last})$ located along a lane with heading direction towards \emph{negative} $\eta$. The path length is much shorter for the first scenario.}
\label{fig_Ex1and2_backtoPos}
\end{figure}

The importance of a normalised coordinate system is stressed in which a  start position $Z_0$ is located as shown in Fig. \ref{fig_CIRCULAR}. In fact, the path planning method CIRC is tailored to such coordinate system representation. Fig. \ref{fig_Ex1and2_backtoPos} illustrates two possible scenarios for resuming work after refilling at the depot. In these two scenarios, $Z(\tau^\text{last})$ is situated either along a lane with heading direction towards positive or negative $\eta$. As indicated, the first method is preferable. This is because it avoids the traversal of an entire interior lane without performing actual application work.

\begin{remark}
\label{rmk_CIRC_resumingWork}
If the field entrance is located such that $Z_0\in\mathcal{Z}_0^{(1)}$ with  $\xi_1 < \xi_0\leq \xi_2$, and $Z(\tau^\text{last})$ for the resumption of work is located such the $\xi(\tau^\text{last})>\xi_0$, then lane $\tilde{k}=1$ must be traversed in order to reach the ``lower'' headland path, before proceeding to $Z(\tau^\text{last})$. Instead, if $\xi_{\tilde{k}} < \xi_{0} \leq \xi_{\tilde{k}+1}$ for $\tilde{k}\in\{2,\dots,\lfloor\frac{N}{2}\rfloor\}$, then either lane $\tilde{k}$ or $\tilde{k}-1$ must be traversed. All these scenarios imply an initial movement towards \emph{negative} $\xi$-direction despite $\xi(\tau^\text{last})>\xi_0$. However, since lane $\tilde{k}$ or $\tilde{k}-1$ are the immediate \emph{next} and the \emph{second next} lane to $Z_0$, the corresponding detour with respect to a movement monotonically increasing from $\xi_0$ towards $\xi(\tau^\text{last})$  is always very small. For  $Z_0\in \mathcal{Z}_0^{(2)}$, such a detour does also not occur.
\end{remark}

\subsection{Path Planning based on the Circular Pattern -- CIRC$^\star$ \label{subsec_CIRCstar}}
 
CIRC$^\star$ is a variation of CIRC when modifying the method for headland path coverage. As will be shown, CIRC$^\star$ is the preferred method for both single-run and partial field coverage. First, a remark is made about path planning for \emph{single-run} field coverage based on the methods according to ABp and CIRC, see also Fig. \ref{fig_MEANDER} and \ref{fig_CIRCULAR}. Under the assumption of an initial \emph{uninterrupted} headland path traversal, ABp is preferred over CIRC with respect to path length minimisation. Some headland edges are traversed less frequently. Specifically, for CIRC every second ``lower'' headland edge is traversed three times:  once along the headland path traversal, and two times according to the path planning pattern of Fig. \ref{fig_CIRCULAR_1unit}. Similarly, also every second ``upper'' headland edge is traversed three times: once along the headland path, once according to the pattern of Fig. \ref{fig_CIRCULAR_1unit}, and once after traversal of the $N$th lane when returning to $Z_0$. This frequent traversal of the same edges is suboptimal. However, when dropping the assumption of an initial uninterrupted headland path traversal, an optimal field coverage path can be constructed based on the pattern of Fig. \ref{fig_CIRCULAR_1unit}. Specifically, the headland path is covered as a byproduct of concatenations of the proposed circular pattern. Traversing theses concatenations, every second ``upper'' headland edge is not yet covered. However, after traversal of the final lane, all these edges can be covered when returning to $Z_0$ along the ``upper'' headland path. This method is referred to as CIRC$^\star$ and represents the optimal field coverage method since every edge is covered \emph{at most} twice. For an even $N$, the set of edges that are traversed twice is confined to headland segments. For an odd $N$, the two edges $e_{N,2N+2}$ and $e_{2N+2,2N}$ are additionally traversed twice. See also Fig. \ref{fig_CIRCULARplus} for illustration. 
\begin{proposition}
For field shapes with uninterrupted lanes when aligned in a rotated coordinate frame as displayed in Fig. \ref{fig_2FieldShapes}, CIRC$^\star$ provides the single-run field coverage path plan of minimal total path length.
\end{proposition}
\begin{proof}
For field coverage, all edges of the transition graph must be covered at least once. Therefore, for our setting with odd vertices, a transition graph extension is required by duplicating some edges until an \emph{Eulerian graph}, see \cite{bondy1976graph}, is generated. This augmentation can be conducted in a minimal path length manner by finding corresponding odd-degree node-pairings. Then, an \emph{Eulerian tour} on this graph guarantees traversal of all edges. CIRC$^\star$ provides such Eulerian tour.
\end{proof}
This discussion is likewise relevant for \emph{partial} field coverage, since its overall path length is just composed of the field coverage path length plus the summed distances from returning to the depot and when resuming work in the field. The concepts for returning to the depot and resuming work in the field are identical for CIRC and CIRC$^\star$. Both are based on the same path planning pattern displayed in Fig. \ref{fig_CIRCULAR_1unit}. However, they differ in their method of headland path coverage. This implies different on/off switching sequences for the nozzles of an automatic section control (ASC) system, see~\cite{batte2006economics} and \cite{luck2010potential}. Algorithms 3 and 4 similarly apply for CIRC$^\star$. To summarise, for overall path length minimisation and adoption of the modified method for headland traversal, it is proposed to i) use the method according to CIRC$^\star$ for generation of the field coverage path plan, and ii) conduct the returns to the depot for refilling as discussed. For the field shapes under consideration, this method is optimal for \emph{both} single-run field coverage and partial field coverage.

\begin{figure}
\begin{subfigure}[b]{\linewidth}
\centering%
\begin{tikzpicture}
\draw [black,-{Latex[scale=1.0]}] plot [rounded corners=0.25cm] coordinates {(0.5,-4)(-0.8,-4)(-0.8,-5.9)(0.8,-5.9)(0.8,-4.1)(0,-4.1)(0,-5.8)
(2.4,-5.9)(2.4,-4.1)(1.6,-4.1)(1.6,-5.8)
(4.0,-5.9)(4.0,-4.1)(3.2,-4.1)(3.2,-5.8)
(5.6,-5.87)(5.6,-4.1)(4.8,-4.1)(4.8,-5.8)
(5.65,-5.87)(5.65,-4.05)(4.3,-4.05)
};
\draw[fill=black] (0.5,-4) circle (1pt); 
\node[color=black] (a) at (0.5,-3.7) {$Z_0$};
\draw[dotted,blue,line width=1pt] (2.4, -6.2) -- (2.4, -3.8);
\draw [black,-{Latex[scale=1.0]}] plot [rounded corners=0.25cm] coordinates { (-0.8,-4.5)(-0.8,-5)
};
\node[color=black] (a) at (4.8, -3.7) {towards $Z_0$};
%
\end{tikzpicture}
\caption{Odd $N$.}
\end{subfigure}\\[0pt]
\begin{subfigure}[b]{\linewidth}
\centering%
\begin{tikzpicture}
\draw [black,-{Latex[scale=1.0]}] plot [rounded corners=0.25cm] coordinates {
(8,-4)
(6.7,-4)(6.7,-5.9)(8.3,-5.9)(8.3,-4.1)
(7.5,-4.1)(7.5,-5.8)(9.9,-5.9)(9.9,-4.1)
(9.1,-4.1)(9.1,-5.8)(11.5,-5.9)(11.5,-4.1)
(10.7,-4.1)(10.7,-5.8)(13.1,-5.9)(13.1,-4.1)
(12.3,-4.1)(12.3,-5.8)(13.9,-5.9)
(13.9,-4)(11.9,-4)
};
\draw[fill=black] (8,-4) circle (1pt); 
\node[color=black] (a) at (8,-3.7) {$Z_0$};
\draw[dotted,blue,line width=1pt] (10.3, -6.2) -- (10.3, -3.8);
\draw [black,-{Latex[scale=1.0]}] plot [rounded corners=0.25cm] coordinates { (6.7,-4.5)(6.7,-5)
};
\node[color=black] (a) at (12.45, -3.7) {towards $Z_0$};
%
\end{tikzpicture}
\caption{Even $N$.}
\end{subfigure}
\caption{CIRC$^\star$. There is only a distinction between two cases: odd and even $N$. For CIRC$^\star$, the path planning is identical for $\mathcal{Z}_0^{(1)}$ and $\mathcal{Z}_0^{(2)}$ as defined in Fig. \ref{fig_CIRCULAR}. Thus, for CIRC$^\star$ the set of admissible field entrances within the normalised coordinate system is $\mathcal{Z}_0 = \mathcal{Z}_0^{(1)}\cup \mathcal{Z}_0^{(2)}$. An exemplatory $Z_0\in\mathcal{Z}_0$ is shown. The field coverage path is displayed as not closed in order to visualise the manner in which the headland path is traversed. The blue dotted lines indicate the $\frac{\xi_1 + \xi_N}{2}$-coordinate, respectively.}
\label{fig_CIRCULARplus}
\end{figure}
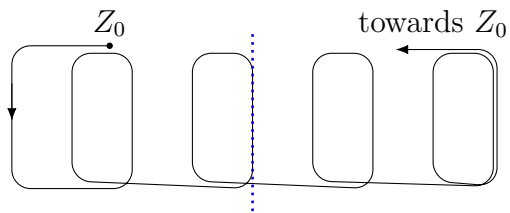
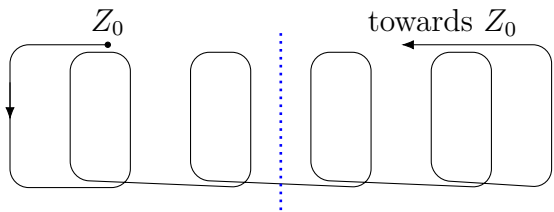

\begin{figure*}
\vspace{0.3cm}
\begin{subfigure}[b]{\linewidth}
\centering%
\begin{tikzpicture}
\draw [black] plot [rounded corners=0.25cm] coordinates { (0.4,2)(-0.8,2)(-0.8,0)(13.6,0)(13.6,2)(0.45,2)
(0,2)(0,0)(0.8,0)(0.8,2)(1.6,2)(1.6,0)
(2.4,0)(2.4,2)(3.2,2)(3.2,0)
(4.0,0)(4.0,2)(4.8,2)(4.8,0)
(5.6,0)(5.6,2)(6.4,2)(6.4,0)
(7.2,0)(7.2,2)(8,2)(8,0)
(8.8,0)(8.8,2)(9.6,2)(9.6,0)
(10.4,0)(10.4,2)(11.2,2)(11.2,0)
(12,0)(12,2)(12.8,2)(12.8,0)
(12.4,0)
};
\draw[fill=black] (0.55,2) circle (1pt); 
\node[color=black] (a) at (0.55, 2.3) {$Z_0$};
%
\draw [black,-{Latex[scale=1.0]}] plot [rounded corners=0.25cm] coordinates { (-0.8,1.5)(-0.8,1)
};
\draw[fill=black] (3.2,0.65) circle (1pt); 
\node[color=black] (a) at (2.8, 0.65) {$Z(t)$};
\draw [draw=blue!50,draw opacity=0.5, line width=4pt,-{Latex[scale=1.0]}] plot [rounded corners=0.25cm] coordinates { (3.2,0.65)(3.2,0)(12.8,0)(12.8,2)(0.55,2)};
\end{tikzpicture}
\caption{ABp.}
\end{subfigure}\\%
\begin{subfigure}[b]{\linewidth}
\centering%
\begin{tikzpicture}
\draw [black] plot [rounded corners=0.25cm] coordinates { (0.4,-3.25)(-0.8,-3.25)(-0.8,-5.25)(13.6,-5.25)(13.6,-3.25)(0.45,-3.25)(0,-3.25)
(0,-5.25)(1.6,-5.25)(1.6,-3.25)(0.8,-3.25)(0.8,-5.25)
(1.6,-5.25)(3.2,-5.25)(3.2,-3.25)(2.4,-3.25)(2.4,-5.25)
(3.2,-5.25)(4.8,-5.25)(4.8,-3.25)(4,-3.25)(4,-5.25)
(4.8,-5.25)(6.4,-5.25)(6.4,-3.25)(5.6,-3.25)(5.6,-5.25)
(6.4,-5.25)(8,-5.25)(8,-3.25)(7.2,-3.25)(7.2,-5.25)
(8,-5.25)(9.6,-5.25)(9.6,-3.25)(8.8,-3.25)(8.8,-5.25)
(8.8,-5.25)(11.2,-5.25)(11.2,-3.25)(10.4,-3.25)(10.4,-5.25)
(10.4,-5.25)(12.8,-5.25)(12.8,-3.25)(12,-3.25)(12,-5.25)
(12.4,-5.25)
};
\draw[fill=black] (0.55,-3.25) circle (1pt); 
\node[color=black] (a) at (0.55, -2.95) {$Z_0$};
%
\draw [black,-{Latex[scale=1.0]}] plot [rounded corners=0.25cm] coordinates { (-0.8,-3.75)(-0.8,-4.25)
};
\draw[fill=black] (3.2,-4.6) circle (1pt); 
\node[color=black] (a) at (2.8, -4.6) {$Z(t)$};
\draw [draw=green,draw opacity=0.5, line width=4pt,-{Latex[scale=1.0]}] plot [rounded corners=0.25cm] coordinates { (3.2,-4.6)(3.2,-3.25)(0.55,-3.25)};
\end{tikzpicture}
\caption{CIRC.}
\end{subfigure}
\caption{Visualisation of a scenario for the example of Section \ref{subsec_parametricEx}. The planned paths for a return from $Z(t)$ to $Z_0$ according to ABp and CIRC are compared. Location $Z(t)$ is identical for both plots (a) and (b). However, because of the planned paths according ABp and CIRC, the initial heading direction along the initial interior lane is different.}
\label{fig_Ex1}
\end{figure*}
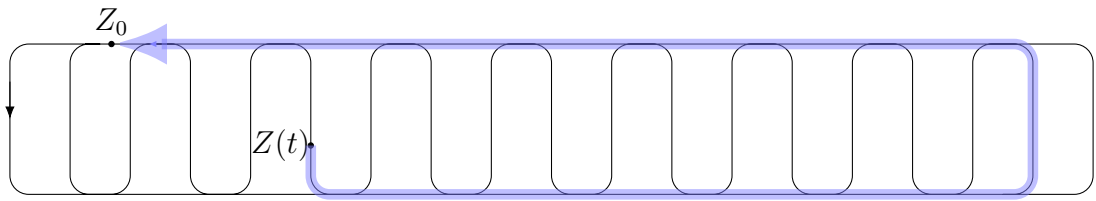
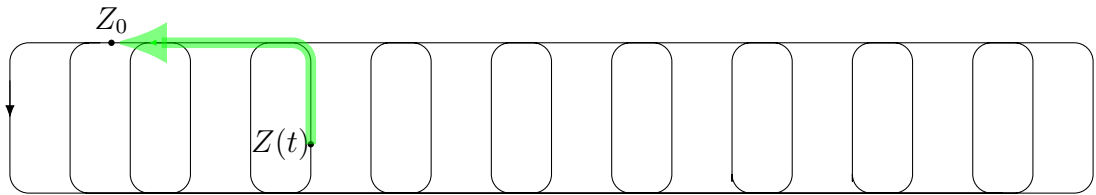

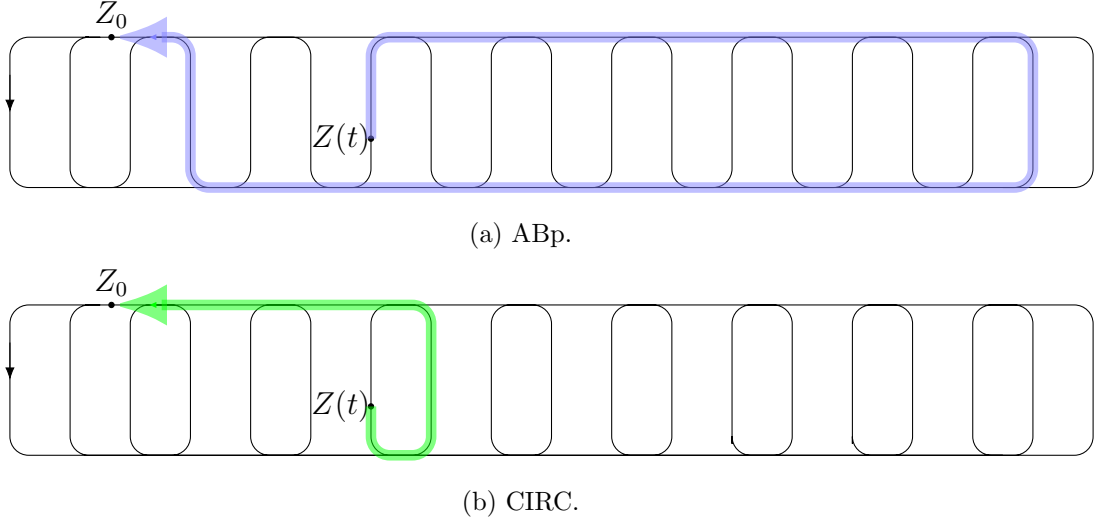
\begin{figure*}[t]
\vspace{0.3cm}
\begin{subfigure}[b]{\linewidth}
\centering%
\begin{tikzpicture}
\draw [black] plot [rounded corners=0.25cm] coordinates { (0.4,2)(-0.8,2)(-0.8,0)(13.6,0)(13.6,2)(0.45,2)
(0,2)(0,0)(0.8,0)(0.8,2)(1.6,2)(1.6,0)
(2.4,0)(2.4,2)(3.2,2)(3.2,0)
(4.0,0)(4.0,2)(4.8,2)(4.8,0)
(5.6,0)(5.6,2)(6.4,2)(6.4,0)
(7.2,0)(7.2,2)(8,2)(8,0)
(8.8,0)(8.8,2)(9.6,2)(9.6,0)
(10.4,0)(10.4,2)(11.2,2)(11.2,0)
(12,0)(12,2)(12.8,2)(12.8,0)
(12.4,0)
};
\draw[fill=black] (0.55,2) circle (1pt); 
\node[color=black] (a) at (0.55, 2.3) {$Z_0$};
%
\draw [black,-{Latex[scale=1.0]}] plot [rounded corners=0.25cm] coordinates { (-0.8,1.5)(-0.8,1)
};
\draw[fill=black] (4,0.65) circle (1pt); 
\node[color=black] (a) at (3.6, 0.65) {$Z(t)$};
\draw [draw=blue!50,draw opacity=0.5, line width=4pt,-{Latex[scale=1.0]}] plot [rounded corners=0.25cm] coordinates { (4,0.65)(4,2)(12.8,2)(12.8,0)(1.6,0)(1.6,2)(0.55,2)};
\end{tikzpicture}
\caption{ABp.}
\end{subfigure}\\%
\begin{subfigure}[b]{\linewidth}
\centering%
\begin{tikzpicture}
\draw [black] plot [rounded corners=0.25cm] coordinates { (0.4,-3.25)(-0.8,-3.25)(-0.8,-5.25)(13.6,-5.25)(13.6,-3.25)(0.45,-3.25)(0,-3.25)
(0,-5.25)(1.6,-5.25)(1.6,-3.25)(0.8,-3.25)(0.8,-5.25)
(1.6,-5.25)(3.2,-5.25)(3.2,-3.25)(2.4,-3.25)(2.4,-5.25)
(3.2,-5.25)(4.8,-5.25)(4.8,-3.25)(4,-3.25)(4,-5.25)
(4.8,-5.25)(6.4,-5.25)(6.4,-3.25)(5.6,-3.25)(5.6,-5.25)
(6.4,-5.25)(8,-5.25)(8,-3.25)(7.2,-3.25)(7.2,-5.25)
(8,-5.25)(9.6,-5.25)(9.6,-3.25)(8.8,-3.25)(8.8,-5.25)
(8.8,-5.25)(11.2,-5.25)(11.2,-3.25)(10.4,-3.25)(10.4,-5.25)
(10.4,-5.25)(12.8,-5.25)(12.8,-3.25)(12,-3.25)(12,-5.25)
(12.4,-5.25)
};
\draw[fill=black] (0.55,-3.25) circle (1pt); 
\node[color=black] (a) at (0.55, -2.95) {$Z_0$};
%
\draw [black,-{Latex[scale=1.0]}] plot [rounded corners=0.25cm] coordinates { (-0.8,-3.75)(-0.8,-4.25)
};
\draw[fill=black] (4,-4.6) circle (1pt); 
\node[color=black] (a) at (3.6, -4.6) {$Z(t)$};
\draw [draw=green,draw opacity=0.5, line width=4pt,-{Latex[scale=1.0]}] plot [rounded corners=0.25cm] coordinates { (4,-4.6)(4,-5.25)(4.8,-5.25)(4.8,-3.25)(0.55,-3.25)};
\end{tikzpicture}
\caption{CIRC.}
\end{subfigure}
\caption{Visualisation of a scenario for the example of Section \ref{subsec_parametricEx}. The planned paths for a return from $Z(t)$ to $Z_0$ according to ABp and CIRC are compared. Location $Z(t)$ is identical for both plots (a) and (b). However, because of the planned paths according ABp and CIRC, the initial heading direction along the initial interior lane is different.}
\label{fig_Ex2}
\end{figure*}

\section{Illustrative Examples\label{sec_IllustrativeEx}}

\subsection{Parametric Example\label{subsec_parametricEx}}

\begin{table}[t]
\vspace{-0.6cm}
\centering
\begin{small}
\begin{tabular}{l|l}
\hline
\rowcolor[gray]{0.8} \multicolumn{2}{c}{\hspace{-0.7cm}Headland path}\\
\hline\hline
\rowcolor[gray]{0.93} from $Z(t)\in e_{j,k}$ to $Z_0$ & $\Delta D_{\text{ABp},\text{CIRC}}$\\[1pt] 
\hline
$j\in\{0,1,N+1,2N+1\}$ & $2(N-3)W_0$ \\
$1<j<N$ and $j$ even & $2(N-j-1)W_0$  \\
$1<j<N$ and $j$ odd & $2(N-j-2)W_0$  \\  
otherwise & $0$ \\[2pt]
\hline
\rowcolor[gray]{0.93} from $Z_0$ to $Z(\tau^\text{last})\in e_{j,k}$ & $\Delta D_{\text{ABp},\text{CIRC}}$\\[1pt] 
\hline
$3 \leq j \leq  N,~j=2N+2$& $ -2q_{l} $ \\
$j=2N+1$ & $0$ \\
$j\geq N+2$ and $j$ even & $-2q_{l} + (4N-2j)W_0-2R$ \\
$j\geq N+2$ and $j$ odd & $-2q_{l} + (4N-2j-2)W_0- 2R$ \\
otherwise & $0$ \\[2pt]
\hline
\rowcolor[gray]{0.8} \multicolumn{2}{c}{\hspace{-0.9cm}Interior lane}\\
\hline\hline
\rowcolor[gray]{0.93} from $Z(t)\in e_{j,N+j}$ to $Z_0$ & $\Delta D_{\text{ABp},\text{CIRC}}$\\[1pt] 
\hline
lane $j=1$ & $2(N-3)W_0$ \\
even lane $j\geq 2$ & $2(1-p)H + 2(N-j-1)W_0 -2R$ \\
odd lane $j\geq 3$ & $2pH + 2(N-j)W_0 -2R$ \\[2pt]
\hline
\rowcolor[gray]{0.93} from $Z_0$ to $Z(\tau^\text{last})\in e_{j,N+j}$ & $\Delta D_{\text{ABp},\text{CIRC}}$\\[1pt] 
\hline
lane $j=1$ & $0$ \\
lane $j=2$ & $-2(1-p)H-2W_0+2R$ \\
even lane $j\geq 4$ & $-2(1-p)H - 2W_0 - 2q_{l}$\\
odd lane $j\geq 3$ & $- 2pH -2q_{l} $ \\[2pt]
\hline
\end{tabular}
\end{small}
\caption{Path length differences for the example of Section \ref{subsec_parametricEx} with odd $N$ and a field entrance $Z_0$ located between the first and second interior lane. For illustration, see also Fig. \ref{fig_Ex1}. It is distinguished between edges along the headland path and edges representing interior lanes. For the former case, index $k$ is determined by the headland path and notation of Fig. \ref{fig_IllustrNotation}. For the latter case, the heading direction of a vehicle along an edge $e_{j,N+j}$ varies for ABp and CIRC according to their field coverage path plans, see Fig. \ref{fig_MEANDER} and \ref{fig_CIRCULAR}. The results for all edges of the entire path network are reported.}
\label{tab_ReturnResume}
\end{table}

\begin{table}
\centering
\begin{small}
\begin{tabular}{l|l}
\hline
\rowcolor[gray]{0.93} Method $m$ & $\Delta D_{\text{ABp},m}$  \\[1pt] 
\hline
CIRC & $-(N-1)W_0$ \\
CIRC$^\star$ & $(N-3)W_0$  \\ 
\hline
\end{tabular}
\end{small}
\caption{Single-run field coverage path lengths for the illustrative example in Section \ref{subsec_parametricEx} with odd $N$. The method according to CIRC has a path length that is proportional to the total number of lanes $N$ \emph{longer} than ABp. In contrast, the method according to CIRC$^\star$ is similarly \emph{shorter} and also scaling linearly with $N$. CIRC$^\star$ and CIRC are both based on the same path planning pattern. However, their method of covering the headland path differs.}
\label{tab_FieldCoverageLength1Run}
\end{table}

For illustration, a rectangular field shape with an odd number of lanes $N$ and a field entrance according to the $\mathcal{Z}_0^{(1)}$-type is considered. $Z_0$ is located between the first and second interior lane. See Fig. \ref{fig_Ex1}, which illustrates one specific scenario discussed further below. The path length differences between ABp and CIRC, which both assume an initial headland path traversal, are reported in Table \ref{tab_ReturnResume} by $\Delta D_{\text{ABp},\text{CIRC}}$. All locations along the entire path network for both the returning to $Z_0$ and the resumption of work at $Z(\tau^\text{last})$ are considered. All locations can be considered since parametric results coincide for locations along the same edges. The nominal distance between ``upper'' and ``lower'' headland path and inter-lane spacing (the machine operating width) are denoted $H_0$ and $W_0$, respectively. Accounting for the turning radius $R$, $H=H_0 - 2R + 2C$ with quarter circle path length $C=\frac{R\pi}{2}$, and $W = W_0-2R$ are defined. The fraction along an interior lane is indicated by $p\in[0,1]$. For example, in Fig. \ref{fig_Ex1} and \ref{fig_Ex2}, $p=0.25$. The lane immediately neighboring $Z_0$ is indicated by $j_l$ with $\xi_{j_{l}}\leq \xi_0< \xi_{j_{l}+1}$. Accordingly, lane $j_{p}$ indicates $\xi_{j_{p}} = \xi_{j_{l}+1}$. Furthermore, $q_l = \xi_0 - \xi_{j_l}$ and $q_p = \xi_{j_p} - \xi_{0}$ are defined such that $q_l + q_{p} = W_0$. 

Several observations can be made from Table \ref{tab_ReturnResume}. First, assuming $p<0.5$ and $2(1-p)H-2R>0$, the largest spread $\Delta D_{\text{ABp},\text{CIRC}}$ is observed for the case of returning from lane $j=2$ to $Z_0$. It results in $\Delta D_{\text{ABp},\text{CIRC}} = 2(1-p)H + 2(N-3)W_0 - 2R$, which grows linearly with the number of lanes $N$. Second, assuming $p<0.5$, the smallest spread is achieved along even lanes with $j>=4$ when resuming work. It results in $\Delta D_{\text{ABp},\text{CIRC}} = -2(1-p)H - 2W_0 - 2q_{l}$, which is negative and therefore indicates a shorter path for ABp. The equivalent setting for the resumption of work along an odd lane $j\geq 3$ is $\Delta D_{\text{ABp},\text{CIRC}} = - 2pH -2q_l $. Third, this motivates guidelines for optimal operation of CIRC with respect to ABp for path length minimisation. Ideally, $q_l\rightarrow 0$ which indicates that $Z_0$ is to be located very close to the first interior lane. Here, the limit-operator is denoted by '$\rightarrow$'. Ideally, a return-command is triggered along a lane with heading towards positive $\eta$ according to the path plan for CIRC and $p\rightarrow 0$. Ideally, the agricultural vehicle permits a small turning radius $R$ and an operating width $W_0>R$. Then, the \emph{only} setting for which the path length for CIRC is worse than for ABp is along lane $N$ (the ultimate lane), which results in $\Delta D_{\text{ABp},\text{CIRC}}\rightarrow -2R$. This is a very small shortcoming since turning radii of agricultural vehicles typically are small (e.g., $R\approx 6$m). In all other settings, i.e., along the remaining entire path network and under the assumptions above, CIRC yields shorter paths for the return to $Z_0$ and the resumption of work in the field. 

In two different scenarios, Fig. \ref{fig_Ex1} and \ref{fig_Ex2} visualise the planned paths for ABp and CIRC for a return from $Z(t)$ to $Z_0$. The two scenarios differ in that $Z(t)$ is located along an odd and even lane, respectively. The assumed vehicle locations upon triggering the return command are identical for ABp and CIRC. However, because of the characteristic path planning for both methods, the vehicle is heading towards different directions along the lanes, respectively. The two scenarios are meant to illustrate the following. First, for ABp the disadvantage of always having to travel until the $N$th lane is made apparent. Second, comparing Fig. \ref{fig_Ex1} and \ref{fig_Ex2}, the benefit of triggering the return command for CIRC on a lane with heading direction towards positive $\eta$ can be observed. The non-working distance, which represents the entire path from $Z(t)$ to $Z_0$, is much shorter for the former scenario. See also Fig. \ref{fig_Ex1and2_backtoPos} for the corresponding resumption of application work.

Table \ref{tab_FieldCoverageLength1Run} indicates the path length differences between the methods ABp, CIRC and CIRC$^\star$ for \emph{single-run} field coverage. CIRC$^\star$ differs from CIRC in its method to handle the headland coverage. ABp and CIRC proceed \emph{sequentially}. Before covering any interior lane, they first cover the headland path entirely. In contrast, CIRC$^\star$ combines lane and headland coverage \emph{simultaneously} as outlined in Section \ref{subsec_CIRCstar}. As Table \ref{tab_FieldCoverageLength1Run} illustrates, ABp is preferable over CIRC with respect to single-run field coverage path length. On the other hand, CIRC is preferable with respect to path length minimisation for returning to $Z_0$ and resuming of work within the field. Thus, dependent on the frequency of such return and resume states, the overall path length for ABp may still be shorter than for CIRC. This, however, changes drastically when employing the method according to CIRC$^\star$. Not only does CIRC$^\star$ enjoy the benefits of the path planning pattern in Fig. \ref{fig_CIRCULAR_1unit} for partial field coverage, it also significantly lowers (linear scaling in $N$) the single-run path length, see Table~\ref{tab_FieldCoverageLength1Run}. Note that even if the compacted area minimisation constraints are neglected, CIRC$^\star$ yields consistently shorter path lengths than ABp.

\subsection{Real-world Example\label{subsec_RealWorldEx}}

\begin{figure}
\vspace{0.1cm}
\centering
\includegraphics[width=14cm]{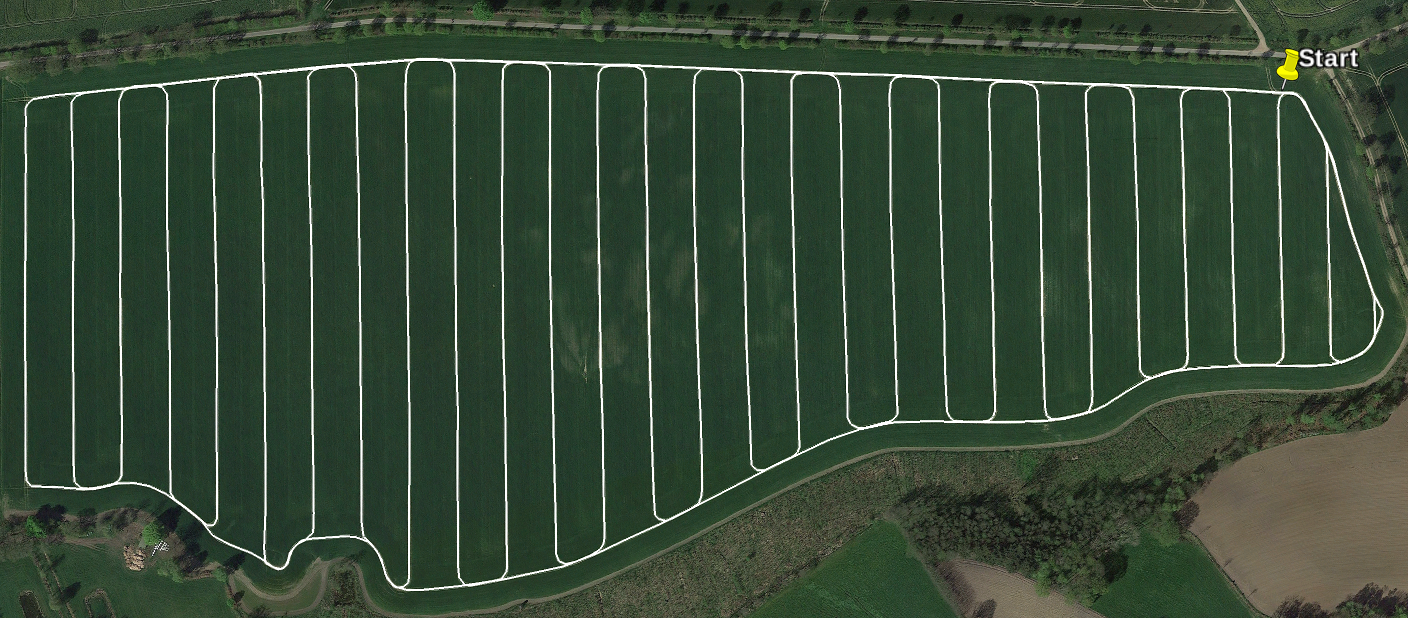}
\caption{Real-world field from Northern Germany ($53^\circ34'30.13''$N, $10^\circ 27'38.14''$E). The current working practice, ABp, is emphasised in white. The working width, turning radius and field size are 36m, 7m, and 32.2ha, respectively. Total field coverage path length according to Algorithm 1 is 12092m.}
\label{fig_field2white}
\end{figure}

\newlength\figureheight
\newlength\figurewidth
\setlength\figureheight{4cm}
\setlength\figurewidth{6cm}
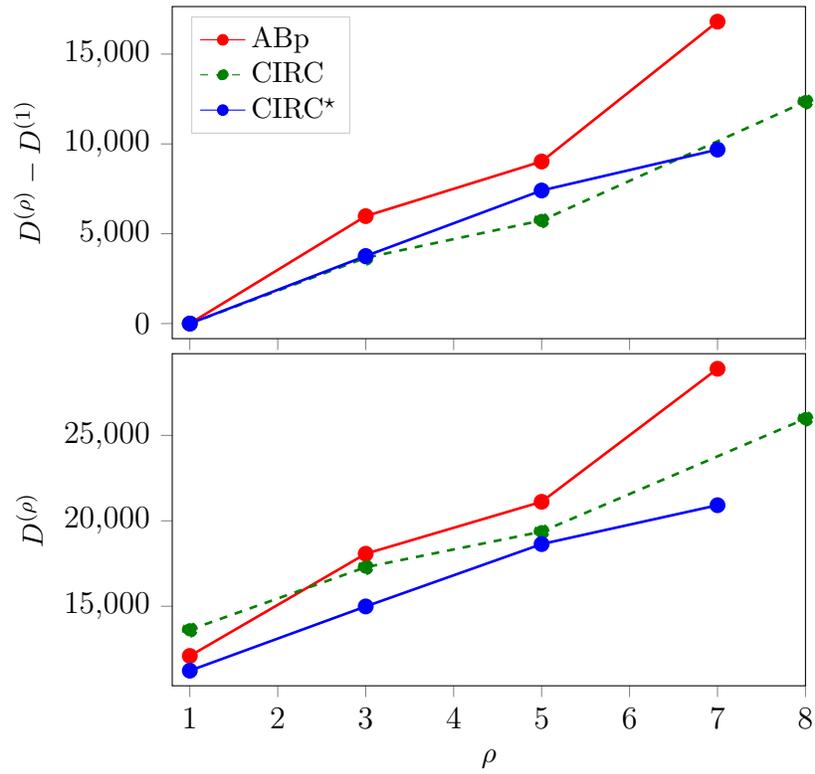
\begin{figure}
\centering
\vspace{0.1cm}
\begin{tikzpicture}

\begin{groupplot}[group style={group size=1 by 2,horizontal sep=1cm, vertical sep=0.2cm}]
\nextgroupplot[
xmin=0.8, xmax=8,
ymin=-840.25, ymax=17645.25,
width=10cm,
height=6cm,
xtick={0,1,2,3,4,5,6,7,8,9},
scaled y ticks = false,
ylabel={\small{$D^{(\rho)}-D^{(1)}$}},
ylabel shift = 0 pt,
xticklabels={},
tick align=outside,
tick pos=left,
x grid style={lightgray!92.026143790849673!black},
y grid style={lightgray!92.026143790849673!black},
legend entries={\small{ABp},\small{CIRC},\small{CIRC$^\star$}},
legend cell align={left},
legend style={at={(0.03,0.97)}, anchor=north west, draw=white!80.0!black}
]
\addlegendimage{mark=*, red}
\addlegendimage{mark=*, dashed,green!50.0!black}
\addlegendimage{mark=*, blue}
\addplot [red, mark=*,solid,line width=1pt, mark size=2.5, mark options={solid}]
table {%
1 0
3 5980
5 9021
7 16805
};
\addplot [green!50.0!black, dashed,line width=1pt, mark=*, mark size=2.5]
table {%
1 0
3 3662
5 5739
8 12350
};
\addplot [blue, mark=*, line width=1pt, mark size=2.5, mark options={solid}]
table {%
1 0
3 3759
5 7407
7 9683
};
\nextgroupplot[
xlabel={\small{$\rho$}},
ylabel={\small{$D^{(\rho)}$}},
xmin=0.8, xmax=8,
ymin=10346.65, ymax=29780.35,
width=10cm,
height=6cm,
tick align=outside,
tick pos=left,
scaled y ticks = false,
x grid style={lightgray!92.026143790849673!black},
y grid style={lightgray!92.026143790849673!black}
]
\addplot [red, mark=*, mark size=2.5,line width=1pt, solid, forget plot]
table {%
1 12092
3 18072
5 21113
7 28897
};
\addplot [green!50.0!black, mark=*,line width=1pt, mark size=2.5, dashed, forget plot]
table {%
1 13625
3 17287
5 19364
8 25975
};
\addplot [blue,solid, mark=*, mark size=2.5,line width=1pt, forget plot]
table {%
1 11230
3 14989
5 18637
7 20913
};
\end{groupplot}

\end{tikzpicture}
\caption{Results for the real-world field displayed in Fig. \ref{fig_field2white}. For each of ABp, CIRC, and CIRC$^\star$, four experiments were conducted, respectively. See Section \ref{subsec_RealWorldEx} for discussion.}
\label{figFieldEx}
\end{figure}


To also provide a quantitative example, the real-world field displayed in Fig. \ref{fig_field2white} is considered. For each of ABp, CIRC and CIRC$^\star$ four experiments are conducted. In the first experiment, path lengths for single-run field coverage are analysed. In the remaining three experiments it is assumed that after every 5000m, 2500m or 1750m of working distance, a return to $Z_0$ is required for refilling before the resumption of work, respectively. The number of field runs and the corresponding total accumulated path length is denoted by $\rho\geq 1$ and $D^{(\rho)}$, respectively. The results are in Fig. \ref{figFieldEx}. Several observations can be made.

First, CIRC$^\star$ performs best throughout all scenarios with respect to total path length minimisation. For single-run field coverage the overall path length is 7.1$\%$ shorter than the  common ABp-solution. For experiment 4, which requires field coverage in 7 parts, CIRC$^\star$ even yields a reduction of 27.6$\%$ in comparison to ABp.

Second, for single-run field coverage, CIRC is 12.7$\%$ longer than the ABp-baseline. However, for field coverage in three parts it already outperforms ABp, and for the fourth experiment, a reduction of 10.1$\%$ results.

Third, the total path length for CIRC$^\star$ is shorter than the CIRC-equivalent in all scenarios. However, for the third experiment, the excess path length w.r.t. the single-run path length, $D^{(5)}-D^{(1)}$, is longer for CIRC$^\star$, see Fig. \ref{figFieldEx}. This indicates that, for this specific experiment, the points where a return to $Z_0$ is triggered are here located more favourably for CIRC.

Fourth, for the last experiment, CIRC requires to divide the field coverage into 8 parts. In contrast, CIRC$^\star$ and ABp only require 7 field runs. Note that the total path length for CIRC is still 10.1$\%$ shorter than then ABp-baseline. Nevertheless, the disadvantage of the initial full headland coverage characteristic for CIRC is exposed. Because of this characteristic, CIRC has the longest total path length for single-run field coverage. As a consequence, for the given experimental setup, CIRC also requires the most field runs.

\section{Discussion\label{sec_discussion}}

In contrast to route planning methods that do not follow any predetermined pattern motif and instead freely optimise the field coverage path, the proposed method is \emph{pattern-based}. The focus on pattern-based path planning was motivated as follows. Freely optimised route plans for non-rectangularly shaped field contours typically result in unintuitive path plans and irregular sequences of lane traversals. While this does not matter for a fully autonomous robot, it is relevant for the case of vehicles driven by human operators. Namely, a navigation guidance application is required. Some practioners prefer to drive according to well-defined and repeatable patterns, even if some additional (limitedly small) detours may be incurred, rather than following a complex route planning. Furthermore, the acquisition of a routing system may also be costly to a farmer and requires access to his or her field data. In contrast, the proposed pattern-based approach can be applied immediately. Naturally, this does not preclude the possibility of also using the resulting paths as references in a two-layered auto-steering framework with reference tracking as the second layer (\cite{backman2012navigation},  \cite{plessen2017reference}).

For partial field coverage, the coordination of a mobile depot and an in-field operating vehicle becomes relevant. Two options are envisioned. The first option includes model-based a priori planning, where field coverage is preplanned by dividing it into partial field coverages. This method is required for the scheduling of the mobile depot. However, it requires a spray prescription map or a similar measure to predict the partial field coverage. An alternative second option is less model-dependent and therefore more practical. It does not preplan how to partition the entire field coverage. Thus, the mobile depot is called according to need, i.e., once the emptying of the storage tank is foreseeable. This method is particularly useful if support units can be summoned quickly, for example, because of a short traveling distance between a stationary depot and the field entrance. The agricultural vehicle is preferably summoned to $Z_0$ for refilling when it is currently traveling along a lane with heading direction towards positive $\eta$, see Fig. \ref{fig_Ex1and2_backtoPos}.

Finally, note that for the field shapes considered in this paper, i.e., with uninterrupted lanes when aligned in a rotated coordinate frame, the application of CIRC$^\star$ results in path length optimal single-run field coverage path plans. This is stressed to underline that the circular path planning pattern element of Fig. \ref{fig_CIRCULAR_1unit} is useful not only for partial but also for single-run field coverage. Hence and to summarize, for the considered field shapes CIRC$^\star$ is the preferred path planning method for \emph{both} single-run field coverage and field coverage in parts. This was further validated in the parametric as well as the real-world example from Sections \ref{subsec_parametricEx} and \ref{subsec_RealWorldEx}, where CIRC$^\star$ outperformed the ABp-baseline and CIRC with respect to total path length minimisation in all scenarios.

\section{Conclusion\label{sec_conclusion}}

Three path planning methods for partial field coverage were discussed: ABp, CIRC and CIRC$^\star$. Compacted area minimisation constraints and field contours that permit path planning based on concatenations of patterns were assumed. A normalised coordinate system description was discussed, which was derived from at most three linear transformations, a rotation and at most two coordinate reflections.

It was illustrated how the overall field coverage path length can be reduced significantly when covering the headland path as a byproduct of concatenations of the proposed pattern in Fig. \ref{fig_CIRCULAR_1unit}. It was stressed how this is relevant for both partial and single-run field coverage. Thus, it was emphasised how CIRC$^\star$ outperforms ABp for single-run field coverage due to its efficient way of covering the headland, thereby minimising non-working distance. At the same time, CIRC$^\star$ enjoys the favourable properties of the proposed circular pattern for partial field coverage. Thus, for the assumed field shapes, CIRC$^\star$ is the optimal path planning method for both single-run and partial field coverage and therefore the preferred method. These findings suggest to replace the currently widespread practice of in-field path planning based on the AB pattern by the method according to CIRC$^\star$.

Subject of future work may include the filling of storage tanks for weight minimisation and avoidance of soil compaction subject to partial field coverage routes. Furthermore, it may be analysed to what extent free-form optimised field coverage routes for arbitrarily shaped fields improve upon the CIRC$^\star$ method. Finally, the real-time prediction of processes in general, and of the storage tank fill-level dynamics in particular, is an exciting and very challening topic. It is prerequisite for fully cooperative field logistics based on time-scheduling. 

\section*{Acknowledgement}

The author would like to thank the two anonymous reviewers and the Associate Editor for the constructive comments, which greatly helped to improve this manuscript.

\bibliographystyle{model5-names}
\bibliography{mybibfile.bib}
\nocite{*}







\end{document}